\documentclass[runningheads]{llncs}
\usepackage{amsmath,amssymb}
\usepackage{tikz,pgfplots}
\usepackage{algorithm}
\usepackage{varwidth}
\usepackage{xcolor}
\usepackage{autoappendix}
\usepackage{smallsec}
\usetikzlibrary{patterns}
\usepackage{caption}
\usepackage{booktabs}
\newtheorem{reduction}{Reduction}
\usepackage{url}

\usepackage{breakurl}
\usepackage[breaklinks]{hyperref}
\usepackage[noend]{algpseudocode}
\usepackage[inline]{enumitem}
\usepackage{times}
\usepackage{hyperref}
\newcommand{\ignore}[1]{}
\newcommand{\maxsat}{\textsc{MaxSAT}}
\newcommand{\sat}{\textsc{SAT}}
\newcommand{\unsat}{\textsc{UnSAT}}
\newcommand{\pathq}{\textsc{Path}}
\newcommand{\cycle}{\textsc{Cycle}}
\newcommand{\sink}{\textsc{Sink}}
\newcommand{\wmaxsat}{\textsc{Weighted MaxSAT}}
\newcommand{\cons}[1]{\textsc {Cons}(#1)}
\newcommand{\certainty}[1]{\textsc {Certainty}(#1)}
\definecolor{amber}{rgb}{1.0, 0.75, 0.0}
\definecolor{capri}{rgb}{0.0, 0.75, 1.0}
\definecolor{coolblack}{rgb}{0.0, 0.18, 0.39}
\definecolor{asparagus}{rgb}{0.53, 0.66, 0.42}
\begin{document}
\pgfplotsset{every tick label/.append style={font=\scriptsize}}
\setlength{\abovedisplayskip}{5pt}
\setlength{\belowdisplayskip}{5pt}
\setlength{\abovedisplayshortskip}{5pt}
\setlength{\belowdisplayshortskip}{5pt}

\title{A SAT-based System for Consistent Query Answering}
%
%
\author{Akhil A. Dixit\inst{1}\and Phokion G. Kolaitis\inst{1, 2}}
\authorrunning{A. Dixit and Ph. Kolaitis}
%
\institute{University of California Santa Cruz \and IBM Research - Almaden}
\maketitle              
\begin{abstract}
An inconsistent database is a database that violates one or more integrity constraints, such as functional dependencies. Consistent Query Answering is a rigorous and principled approach to the semantics of queries posed against inconsistent databases. The consistent answers to a query on an inconsistent database is the intersection of the answers to the query on every repair, i.e., on every consistent database that differs from the given inconsistent one in a minimal way. Computing the consistent answers of a fixed conjunctive query on a given inconsistent database can be a coNP-hard problem, even though every fixed conjunctive query is efficiently computable on a given consistent database.

We designed, implemented, and evaluated  CAvSAT, a SAT-based system for consistent query answering.  CAvSAT leverages a set of natural reductions from the complement of consistent query answering to SAT and to Weighted MaxSAT. The system is capable of handling unions of conjunctive queries and arbitrary denial constraints, which include functional dependencies as a special case. We report results from experiments evaluating CAvSAT on both synthetic and real-world databases. These results provide evidence that a SAT-based approach can give rise to a comprehensive and scalable system for consistent query answering.
\end{abstract}

\section{Introduction}
Managing inconsistencies in databases is a challenge that arises in several different contexts.
Data cleaning is the main approach towards managing inconsistent databases (see the survey \cite{DBLP:journals/ftdb/IlyasC15}). In data cleaning, clustering techniques and/or domain knowledge are used to
resolve violations of integrity constraints in a given inconsistent database, thus producing a single consistent database.
This approach, however, is often \emph{ad hoc}; for example, if a person has two different social security numbers in a database, which of the two should be kept?

The framework of database repairs and consistent query answering, introduced by Arenas, Bertossi, and Chomicki \cite{Arenas99}, is an alternative, and arguably more principled, approach to data cleaning.
In contrast to data cleaning,
the inconsistent database is left as is; instead, inconsistencies are handled at query time by considering all
possible repairs of the inconsistent database, where a {\em repair} of an inconsistent database $I$ is a consistent database $J$ that  differs
from $I$ in a ``minimal" way.
The main algorithmic problem in this framework is to compute the
   \emph{consistent answers}  to a
query $q$ on a given database $I$, that is, the tuples that lie
in the intersection of the results of $q$ applied on each repair of $I$ (see the monograph \cite{Bertossi11}).
Computing the consistent answers to a query $q$ on $I$ can be  computationally harder  than evaluating $q$ on $I$, because an inconsistent database may have exponentially many repairs.
By now there is an extensive literature on the computational complexity of the consistent answers for  different classes of constraints and queries \cite{CateFK12,Kolaitis12,Koutris17,Wijsen12,Wijsen13}.
For key constraints (the most common  constraints) and for conjunctive queries (the most frequently asked
queries), the consistent answers appear to exhibit an intriguing trichotomy, namely, the consistent answers of every fixed conjunctive query under key constraints are either first-order rewritable (hence, polynomial-time computable), or are polynomial-time computable but not first-order rewritable, or are coNP-complete. So far, this trichotomy has been proved for self-join free conjunctive queries by Koutris and Wijsen \cite{Koutris16,Koutris17}.
Moreover, Koutris and Wijsen designed a quadratic algorithm that, given such a conjunctive query and a set of key constraints, determines the side of the trichotomy in which the consistent answers to the query fall. Prior to this work, Fuxman and Miller identified a class of conjunctive queries, called $C_\textit{forest}$, whose consistent answers are FO-rewritable \cite{Fuxman05,FuxmanM07}. Membership in $C_\textit{forest}$, however, is sufficient but not necessary condition for the FO-rewritability of the consistent answers.

Several academic prototype systems for consistent query answering have been developed \cite{Arenas03,Barcelo03,ChomickiH04,Fuxman05,FuxmanM05,Greco03,Kolaitis13,MannaRT11,MarileoB10}. In particular, the ConQuer system \cite{Fuxman05,FuxmanM05} is tailored to queries in the class $C_\textit{forest}$. 
Other systems use logic programming \cite{Barcelo03,Greco03}, compact representations of repairs \cite{Chomicki04}, or reductions to solvers. Specifically, the system in \cite{MannaRT11} uses reductions to answer set programming, while the EQUIP system in \cite{Kolaitis13} uses reductions to binary integer programming and the subsequent deployment of CPLEX.
It is fair to say, however, no comprehensive and scalable system for consistent query answering exists at present; this state of affairs
 has impeded the broader adoption of the framework of repairs and consistent answers as a principled alternative to data cleaning.

In this paper, we report on a SAT-based system for consistent query answering, which we call CAvSAT (Consistent Answers via \sat{}). The CAvSAT system  leverages  natural reductions from the complement of consistent query answering to \sat{} and to \wmaxsat{}. As such, it can handle the consistent answers to unions of conjunctive queries under \emph{denial} constraints, a broad class of integrity constraints that include functional dependencies (hence also key constraints) as a special case. CAvSAT is the first SAT-based system for consistent query answering.
We carried out a preliminary stand-alone evaluation of CAvSAT on both
synthetic and real-world databases. The first set of experiments involved the consistent answers of conjunctive queries under key constraints on synthetic databases in which each relation has up to one million tuples. 
One of the \emph{a priori} unexpected findings is that, for conjunctive queries whose consistent answers are first-order rewritable, CAvSAT had comparable or even better performance to evaluating the first-order rewritings using a database engine, such as PostgreSQL.
The second set of experiments involved the consistent answers of (unions of) conjunctive queries under functional dependencies on restaurant inspection records in Chicago and New York with some of the relations exceeding 200000 tuples. The CAvSAT source code is available at the GitHub repository \url{https://github.com/uccross/cavsat} via a BSD-style license.

While much more work remains to be done, the experimental finding reported here provide evidence that a SAT-based approach can indeed give rise to a comprehensive and scalable system for consistent query answering.

\section{Basic Notions and Background}


\noindent{\bf Databases, Constraints, and Queries}~
A \textit{relational database schema} $\mathcal{R}$ is a finite collection of relation symbols, each with a fixed positive integer as its arity. The attributes of a relation symbol are names for its columns; attributes can also be identified by their positions, thus  $Attr(R) = \{1, ..., n\}$ denotes the set of  attributes of $R$.
An $\mathcal{R}$-\emph{database instance}  or, simply, an $\mathcal{R}$-\emph{instance}  is a collection $I$ of finite relations $R^I$, one for each relation symbol $R$ in $\mathcal R$. An expression of the form $R^I (a_1, ..., a_n)$ is a \textit{fact} of the instance $I$ if $(a_1, ..., a_n) \in R^I$.
Every $\mathcal R$-instance can be identified with the (finite) set of its facts. The \emph{active domain} of $I$ is the set of all values occurring in facts of $I$.

Relational database schemas are often accompanied by a set of integrity constraints that impose semantic restrictions on the allowable instances.
A \textit{functional dependency} (FD) $ \vec{x} \rightarrow \vec{y}$ on a relation symbol $R$  is an integrity constraint asserting that if two facts  agree on the attributes in $\vec{x}$, then they must also agree on the attributes in $\vec{y}$.  A \textit{key} is a minimal subset $\vec{x}$ of $Attr(R)$ such that the FD  $ \vec{x} \rightarrow Attr(R)$ holds. In this case, the attributes in $\vec{x}$ are called \textit{key attributes} of $R$ and they are denoted by underlining their corresponding positions; thus, $R(\underline{A, B}, C)$ denotes that the attributes $A$ and $B$ form a key of $R$. Every functional dependency is expressible in first-order logic. For example, the key constraint $A,B \rightarrow C$ in  $R(\underline{A, B}, C)$ is expressed by the first-order formula
$$\forall x, y, z,  z' (R(x, y, z) \land R(x, y, z') \rightarrow  z = z')$$

Functional dependencies are an important special case of \emph{denial constraints} (DCs), which are expressible by first-order formulas of the form
$$\forall x_1, ..., x_n \neg (\varphi(x_1, ..., x_n) \land \psi(x_1, ..., x_n)),$$
or, equivalently,
$$\forall x_1, ..., x_n  (\varphi(x_1, ..., x_n) \rightarrow \neg \psi(x_1, ..., x_n)),$$
where $\varphi(x_1, ..., x_n)$ is a conjunction of atomic formulas and $\psi(x_1, ..., x_n)$ is a conjunction of  expressions of the form $(x_i\; \mbox{op}\; x_j)$ with each $\mbox{op}$  a built-in predicate, such as $=, \neq, <, >, \leq, \geq$.
%
In words, a denial constraint prohibits a set of tuples that satisfy certain conditions from appearing together in a database instance.

Let $k$ be a positive integer.  A \emph{$k$-ary query} on a relational database schema $\mathcal R$ is a function $q$ that takes  an $\mathcal R$-instance $I$ as argument and returns a $k$-relation $q(I)$ on the active domain of $I$ as value.  A \emph{boolean query} on $\mathcal R$ is a function that takes an $\mathcal R$-instance $I$ as argument and returns true or false as value. As is well known, first-order logic has been successfully used as a query language. In fact, it forms the core of SQL,  the main commercial database query language.

A \textit{conjunctive query} is a first-order formula built using the relational symbols, conjunctions, and existential quantifiers. Thus, each conjunctive query is expressible by a first-order formula of the form
$$q(\vec{z}):= \exists \vec{w}\; (R_1(\vec{x_1}) \land ... \land R_m(\vec{x_m})),$$
where each $\vec{x_i}$ is a tuple consisting of variables and constants, $\vec{z}$ and $\vec{w}$ are tuples of variables, and the variables in $\vec{x_1}, ..., \vec{x_m}$ appear in exactly one of $\vec{z}$ and $\vec{w}$. Clearly, a conjunctive query with $k$ free variables $\vec{z}$ is a $k$-ary query,
while a conjunctive query with no free variables (i.e., all variables are existentially quantified) is a boolean query.
Conjunctive queries are also known as \emph{select-project-join} (SPJ) queries and  are among the most frequently asked queries in databases. For example, the binary  conjunctive query $q(s,t):= \exists c (\mathrm{Enrolls}(s,c) \land \mathrm{Teaches}(t,c))$ returns the set of all pairs $(s,t)$ such that  student $s$ is enrolled in a course taught by  teacher $t$, while the boolean conjunctive query $q():= \exists x, y, z (E(x,y)\land E(y,z)\land E(z,x))$ tests whether or not a graph with an edge relation $E$ contains a triangle.

\medskip

\noindent{\bf Repairs and Consistent Answers}~ Let $\mathcal{R}$ be a database schema and let $\Sigma$ be a set of integrity constraints  on $\mathcal{R}$.
An $\mathcal R$-instance $I$ is \emph{consistent} if $I \models \Sigma$, that is, $I$ satisfies every constraint in $\Sigma$; otherwise, $I$ is  \emph{inconsistent}.
 A \emph{repair} of an inconsistent  instance $I$ w.r.t. $\Sigma$  is a  consistent instance $J$ that differs from $I$ in a ``minimal" way.  Different notions of minimality give rise to different types of repairs (see \cite{Bertossi11} for a comprehensive survey). Here, we  focus on \emph{subset repairs}, the most extensively studied type of repairs. An instance $J$ is  a \emph{subset repair} of an  instance $I$ if $J \subseteq I$ (where $I$ and $J$ are viewed as sets of facts), $J\models \Sigma$,  and there exists no instance $J'$ such that $J'\models \Sigma$ and $J \subset J'\subset I$. From now on, by \emph{repair} we mean a subset repair.
  Arenas, Bertossi, and Chomicki \cite{Arenas99} used  repairs  to give rigorous semantics to query answering on inconsistent databases. Specifically, assume that $q$ is a query,  $I$ is an $\mathcal R$-instance, and $\vec{t}$ is a tuple of values. We  say that $\vec{t}$ is a \emph{consistent answer} (also referred as a \emph{certain answer}) to $q$ on $I$ w.r.t. $\Sigma$ if $\vec{t} \in q(J)$, for every repair $J$ of $I$. We write $\cons{q, I, \Sigma}$ to denote the set of all \emph{consistent answers} to $q$ on  $I$ w.r.t. $\Sigma$, i.e.,
  $$\cons{q, I, \Sigma} = \bigcap\{q(J): \mbox{$J$ is a repair of $I$ w.r.t. $\Sigma$}\}.$$

   If $\Sigma$ is a fixed set of integrity constraints and $q$ is a fixed query, then the main computational problem associated with the consistent answers is:
   given an instance $I$,  compute \cons{$q$, $I$, $\Sigma$}.
If $q$ is a boolean query, then  computing the certain answers becomes the decision problem $\certainty{q,\Sigma}$: given an instance $I$,
is $q$ true on every repair $J$ of $I$ w.r.t.$ \Sigma$? When the constraints $\Sigma$ are understood from the context, we will write $\cons{q, I}$
and $\certainty{q}$, instead of $\cons{q, I,\Sigma}$
and $\certainty{q,\Sigma}$.


\medskip

\noindent{\bf Computational Complexity of Consistent Answers}
 If $\Sigma$ is a fixed finite set of denial constraints and  $q$ is a $k$-ary conjunctive query, where $k\geq 1$, then the following problem is in coNP: given an instance  $I$ and a tuple $\vec{t}$, is $\vec{t}$ a certain answer to $q$ on $I$ w.r.t.\ $\Sigma$?  This is so because to check that $\vec{t}$ is not a certain answer to $q$ on $I$ w.r.t.\ $\Sigma$, we  guess a repair $J$ of $I$ and verify that $\vec{t}\not \in q(J)$ (note that $J$ is a subset of $I$, evaluating a fixed conjunctive query on a given database is a polynomial-time task, and testing if $J$ is a repair of $I$ w.r.t. denial constraints is a polynomial-time task as well). Similarly, if $q$ is a boolean conjunctive query, then the decision problem $\certainty{q,\Sigma}$ is in coNP.
%

Even for key constraints and boolean conjunctive queries, $\certainty{q,\Sigma}$ exhibits a variety of behaviors within coNP. Indeed, consider the queries
\begin{enumerate}
\item $\pathq():= \exists x, y, z\; R(\underline{x}, y) \land S(\underline{y}, z)$;
\item $\cycle():= \exists x, y\; R(\underline{x}, y) \land S(\underline{y}, x)$;
\item $\sink():= \exists x, y, z\; R(\underline{x}, z) \land S(\underline{y}, z)$.
\end{enumerate}
Fuxman and Miller \cite{FuxmanM07} showed that $\certainty{\pathq}$ is FO-rewritable, i.e., there is a first-order definable boolean query $q'$ such that  $\cons{\pathq,\Sigma,I}=q'(I)$, for every instance $I$.
 In fact, $q'$ is $\exists x, y, z\; R(\underline{x}, y) \land S(\underline{y}, z) \land \forall y'(R(\underline{x}, y') \rightarrow \exists z' S(\underline{y'}, z')).$
 Wijsen \cite{WijsenR10} showed that  $\certainty{\cycle}$ is in P, but it is not
FO-rewritable, while Fuxman and Miller \cite{FuxmanM07} showed that $\certainty{\sink}$ is coNP-complete via a reduction from the complement of
\textsc{Monotone 3-SAT}.

The preceding state of affairs sparked a series of investigations aiming to obtain classification results concerning the computational complexity of the consistent answers (e.g., see  \cite{Grieco05,Kolaitis12,Lembo06,Wijsen09,WijsenR10}). The most definitive result to date is a \emph{trichotomy} theorem,  established by Koutris and Wijsen \cite{Koutris15,Koutris16,Koutris17}, for boolean self-join free conjunctive queries, where a conjunctive query is \emph{self-join free} if no relation symbol occurs more than once in the query. This trichotomy theorem asserts that if $q$ is a self-join free conjunctive query with one key per relation symbol, then $\certainty{q}$ is FO-rewritable, or in P but not FO-rewritable, or coNP-complete. Moreover, there is a quadratic algorithm to decide, given such a query, which of the three cases of the trichotomy holds. It remains an open problem whether or not   this trichotomy extends to arbitrary boolean conjunctive queries and  to  arbitrary functional dependencies or denial constraints.

\medskip


\section{Consistent Query Answering for Key Constraints}\label{sec:key-constraints}
In this section, we assume that $\mathcal{R}$ is a database schema and  $\Sigma$ is a finite set of primary key constraints on $\mathcal R$, i.e., there is one key constraint per each relation of $\mathcal{R}$.
We first consider boolean conjunctive queries and,  for each fixed boolean conjunctive query $q$,
we  give a natural polynomial-time reduction from \certainty{$q$} to \unsat{}. We then extend this reduction to non-boolean conjunctive queries, so that for every fixed non-boolean conjunctive query $q$, the consistent answers to $q$ can  be computed by iteratively solving \wmaxsat{} instances. In what follows, we heavily use the notions of \textit{key-equal groups} of facts and  \textit{minimal witnesses} to a conjunctive query.
\begin{definition}\label{key-equal-group}\textbf{Key-Equal Group.} \emph{Let $I$ be an $\mathcal{R}$-instance. We say that two facts of a relation $R$ of $I$ are \textit{key-equal}, if they agree on the key attributes of $R$. A set $S$ of facts of $I$ is called a \textit{key-equal group} of facts if every two facts in $S$ are key-equal, and no fact in $S$ is key-equal to some fact in  $I\backslash S$.}
\end{definition}
\begin{definition}\label{min-witness}\textbf{Minimal Witness.} \emph{Let $I$ be an $\mathcal{R}$-instance and let $S$ be a sub-instance of $I$. We say that $S$ is a \textit{minimal witness} to a conjunctive query $q$ on $I$, if $S \models q$, and for every proper subset $S'$ of $S$, we have that $S' \not\models q$.}
\end{definition}

For each relation $R$ of $I$, the key-equal groups of $R$ are computed by an SQL query that involves grouping the key attributes of $R$. Similarly, the set of minimal witnesses to a fixed conjunctive query $q$ on $I$ are computed efficiently as follows. A unique integer \textit{factID} is attached to each fact, by adding an attribute \textit{FactID} to each relation in $I$ that appears in $q$. Thus, a new instance $I'$ is built, where each relation $R'(\textit{FactID}, \vec{\underline{A}}, \vec{B})$ in $I'$ is obtained from a relation $R(\vec{\underline{A}}, \vec{B})$ in $I$. A new non-boolean query $q'$ is constructed, such that each atom $R'(\textit{factID}_R, \vec{\underline{x}}, \vec{y})$ in $q'$ is constructed from an atom $R(\vec{\underline{x}}, \vec{y})$ of $q$. The variables of $q'$ that correspond to the \textit{FactID} attributes are not existentially quantified. It is easy to see that each tuple (without duplicate \textit{factID}s) in $q'(I')$ is in 1-1 correspondence with a minimal witness to $q$ on $I$.

\newpage
 \noindent{\bf Boolean Conjunctive Queries}
 Let  $q$ be  a fixed boolean conjunctive query over $\mathcal{R}$.

\begin{reduction}\label{reduction1}Given an $\mathcal R$-instance $I$, we  construct a CNF-formula $\phi$  as follows.

For each fact $f_i$ of $I$,  introduce a boolean variable $x_i$, $1\leq i\leq n$.
  Let $\mathcal{G}$ be the set of key-equal groups of facts of $I$, and let $\mathcal{W}$ be the set of minimal witnesses to $q$ on $I$.
\begin{itemize}
\item For each $G_j \in \mathcal{G}$, construct the clause $\alpha_j = \underset{f_i \in G_j}{\lor} x_i$.
\item For each $W_j \in \mathcal{W}$, construct the clause $\beta_j = \underset{f_i \in W_j}{\lor} \neg x_i$.
\item Construct the boolean formula $\phi = \bigg(\overset{|\mathcal{G}|}{\underset{i=1}{\land}}\alpha_i \bigg) \land \bigg(\overset{|\mathcal{W}|}{\underset{j=1}{\land}}\beta_j$\bigg).

\end{itemize}
\end{reduction}
\preservecounter{proposition}
\movetoappendix{\subsection*{Proof of Proposition~\ref{prop1}}}
\copytoappendix{
\begin{proposition}\label{prop1}
Let $\phi$ be the CNF-formula constructed using Reduction \ref{reduction1}.
\begin{itemize}
\item The size of $\phi$ is polynomial in the size of $I$.
\item The formula $\phi$ is satisfiable if and only if \certainty{$q$, $\Sigma$} is false on $I$.
\end{itemize}
\end{proposition}
}
The proofs of all propositions  are given in the Appendix.

\movetoappendix{
\begin{proof}
Let $n$ be the number of facts in $I$. There are exactly $n$ boolean variables used in $\phi$. Clearly, $|\mathcal{G}| \leq n$, therefore the number of $\alpha$-clauses is bounded above by $n$. Similarly, for each $G_j \in \mathcal{G}$, we have that $|G_j| \leq n$. Hence, the length of each $\alpha$-clause is also at most $n$. If $d$ is the number of atoms in $q$, then we have that $|\mathcal{W}| \leq n^d$; moreover, for every $W_j \in \mathcal{W}$, we have that $|W_j| \leq d$. Hence, the number of $\beta$-clauses in $\phi$ is at most $n^d$, and the length of each $\beta$-clause is bounded above by $d$. Since the query $q$ is not part of the input to the problem \certainty{$q$}, we have that $d$ is a fixed constant.

To prove the second part of the proposition,  assume first that $\certainty{q}$ is false on $I$. Hence, there exists a repair $R$ of $I$ that falsifies $q$. Construct an assignment $\hat{a}$ to the variables in $\phi$ by setting $\hat{a}(x_i) = 1$ if and only if $f_i \in R$. Since exactly one fact from each key-equal group of $I$ is present in $R$, exactly one variable from each $\alpha$-clause is set to 1 in $\hat{a}$. Also, since $R \not\models q$, no minimal witness to $q$ is in $R$. Therefore, at least one variable from each $\beta$-clause is set to 0 in $\hat{a}$. Hence, $\hat{a}$ satisfies $\phi$. For the other direction, let $\hat{a}$ be a satisfying assignment to $\phi$. Since no two $\alpha$-clauses share a variable, we can construct a set $X$ of variables by arbitrarily choosing exactly one $x_i$ from each $\alpha$-clause, such that $\hat{a}(x_i) = 1$. Construct a set $R$ of facts of $I$, such that $f_i \in R$ if and only if $x_i \in X$. It is easy to see that $R$ contains exactly one fact from each key-equal group of $I$, and no minimal witness to $q$ on $I$ is present in $R$. Hence, $R$ is a repair of $I$ that falsifies $q$.
\end{proof}
}

\medskip
\noindent {\bf Non-boolean Conjunctive Queries}
 Let $q$ be a fixed non-boolean query on $\mathcal R$, i.e., $q$ has one or more free variables. We extend Reduction \ref{reduction1} to Reduction \ref{reduction2}, so that one can reason about the certain answers to
  $q$ on an $\mathcal R$-instance $I$ using
  the satisfying assignments of the CNF-formula $\phi$ constructed via Reduction \ref{reduction2}.

 \looseness = -1 We  use the term \emph{potential answers} to refer to the answers to $q$ on $I$.
  If $\vec{a}_l$ is such a potential answer, we write   $q[\vec{a}_l]$ to denote the boolean conjunctive query obtained from $q$ by replacing the free variables in the body  of $q$ by corresponding constants from $\vec{a}_l$.

\begin{reduction}\label{reduction2}\begin{sloppypar}
Given an $\mathcal R$-instance $I$, we construct a CNF-formula $\phi$ as follows.

For each fact $f_i$ of $I$,  introduce a boolean variable $x_i$, $1\leq i\leq n$,
Let $\mathcal{G}$ be the set of key-equal groups of facts of $I$ and   let $\mathcal{A}$ be the set of potential answer to $q$ on $I$. For each $\vec{a}_l \in \mathcal{A}$, let $\mathcal{W}^l$ denote the set of minimal witnesses to the boolean query $q[\vec{a}_l]$ on $I$. For each $\vec{a}_l \in \mathcal{A}$, introduce a boolean variable $p_1$,
$1\leq l \leq  ..., |\mathcal{A}|$.
\begin{itemize}
\item For each $G_j \in \mathcal{G}$, construct the clause $\alpha_j = \underset{f_i \in G_j}{\lor} x_i$.
\item For each $\vec{a}_l \in \mathcal{A}$ and for each $W^l_j \in \mathcal{W}^l$, construct the clause ${\beta^l_j = \bigg(\underset{f_i \in W^l_j}{\lor} \neg x_i \bigg) \lor \neg p_l}$.
\item Construct the boolean formula $\phi = \bigg(\overset{|\mathcal{G}|}{\underset{i=1}{\land}}\alpha_i\bigg)\land\bigg(\overset{|\mathcal{A}|}{\underset{l=1}{\land}} \bigg( \overset{|\mathcal{W}^l|}{\underset{j=1}{\land}}\beta^l_j\bigg)\bigg)$.
\end{itemize}\end{sloppypar}
\end{reduction}

\preservecounter{proposition}
\movetoappendix{\subsection*{Proof of Proposition~\ref{prop2}}}
\copytoappendix{
\begin{proposition}\label{prop2}
Let $\phi$ be the CNF-formula constructed using Reduction \ref{reduction2}.
\begin{itemize}
\item The size of $\phi$ is polynomial in the size $I$.
\item There exists a satisfying assignment to $\phi$ in which a variable $p_l$ is set to 1 if and only if $\vec{a}_l \notin \cons{q, I}$.
\end{itemize}
\end{proposition}
}
\movetoappendix{
\begin{proof}Let $n$ be the number of facts in $I$. Let $m$ be  the arity of the query $q$ and let  $d$ and the number of atoms of  $q$. Since an answer to $q$ is a set of $m$ facts, we have that $|\mathcal{A}| \leq n^m$. For each $l$, the number of witnesses in $\mathcal{W}^l$ is bounded by $n^d$. Therefore, there are at most $n^{m + d}$ $\beta$-clauses in $\phi$, each of length at most $d$. Since the query $q$ is not part of the input to \cons{$q$}, the quantities $m$ and $d$ can be treated are fixed constants. It follows directly from Proposition \ref{prop1} that both the number of $\alpha$-clauses and the length of each $\alpha$-clause in $\phi$ are  bounded above by $n$.

To prove the second part of the proposition, assume first that $\vec{a}_l \notin \cons{q}$. Hence, there exists a repair $R$ of $I$, such that no minimal witness to $q[\vec{a}_l]$ is in $R$. Construct an assignment $\hat{a}$ to the variables in $\phi$ as follows. Set $\hat{a}(x_i) = 1$ if and only if $f_i \in R$. Set $\hat{a}(p_l) = 1$, and set $\hat{a}(p_j) = 0$ for all $j \neq l$. Since exactly one fact from each key-equal group of $I$ is in $R$, the assignment sets to 1 exactly one variable from each $\alpha$-clause. Since no minimal witness to $q[\vec{a}_l]$ is in $R$, at least one variable from each $\beta^l$-clause is set to 0 in $\hat{a}$, thus satisfying all $\beta^l$-clauses, even when $p_l$ is set to 1. All other $\beta$-clauses are satisfied trivially because of the assignment $\hat{a}(p_j) = 0$, for all $j \neq l$. In the other direction, let $\hat{a}$ be the satisfying assignment to $\phi$, such that $\hat{a}(p_l) = 1$. Since no two $\alpha$-clauses share a variable, we can construct a set $X$ of variables by arbitrarily choosing exactly one $x_i$ from each $\alpha$-clause, such that $\hat{a}(x_i) = 1$. Construct a set $R$ of facts of $I$, such that $f_i \in R$ if and only if $x_i \in X$. It is easy to see that exactly one fact from each key-equal group of $I$ is present in $R$. Since $\hat{a}(p_l) = 1$  and since all $\beta^l$-clauses are satisfied by $\hat{a}$, at least one fact from each minimal witness to $q[\vec{a}_l]$ is missing in $R$. Hence, $R$ must be a repair of $I$ such that $R \not\models q[\vec{a}_l]$.
\end{proof}
}
\begin{example}\label{example1}
Consider the flights information database in Table \ref{flights}. The database schema has three relations, namely, \textit{Airlines}, \textit{Tickets}, and \textit{Flights}; the key attributes of each relation are underlined.
This database is inconsistent, as the sets $\{f_1, f_3\}$ and $\{f_8, f_9\}$ of facts violate the key constraints of the relations \textit{Airlines} and \textit{Flights}, respectively.

Suppose we want to find out the codes of the flights that belong to an airline from Canada and   fly to the airport OAK. This can be expressed by the unary conjunctive query $q(x):= \textit{Flights}(x,y,z,p,\text{`OAK'},q,r) \land \textit{Airlines}(z, \text{`Canada'})$.
\begin{table}
\caption{Flight information records.}\label{flights}
\begingroup
\renewcommand{\arraystretch}{1.1}
\setlength{\tabcolsep}{3pt}
\begin{minipage}[t]{0.4\textwidth}
\begin{tabular}[t]{|c||c|c|}
\hline
\multicolumn{3}{|c|}{\textit{Airlines}}\\\hline
\textit{Fact} & \underline{AIRLINE} & COUNTRY\\
\hline\hline
$f_1$ & Southwest & United States\\\hline
$f_2$ & Jazz Air & Canada\\\hline
$f_3$ & Southwest & Canada\\\hline
\end{tabular}
\end{minipage}
\begin{minipage}[t]{0.55\textwidth}
\hspace{-0.28in}
\begin{tabular}[t]{|c||c|c|c|c|}
\hline
\multicolumn{5}{|c|}{\textit{Tickets}}\\\hline
\textit{Fact} & \underline{PNR} & CODE & CLASS & FARE\\\hline\hline
$f_4$ & MJ9C8R & SWA 1568 & Economy & 430 USD\\\hline
$f_5$ & KLF88V & MI 471 & First & 914 USD\\\hline
$f_6$ & NJ5RT3 & SWA 1568 & First & 112 USD\\\hline
\end{tabular}
\end{minipage}\vspace{0.02in}
\begin{minipage}{\textwidth}
\begin{tabular}[t]{|c||c|c|c|c|c|c|c|}
\hline
\multicolumn{8}{|c|}{\textit{Flights}}\\\hline
\textit{Fact} & \underline{CODE} & \underline{DATE} & AIRLINE & FROM & TO & DEPARTURE & ARRIVAL\\
\hline\hline
$f_7$ & JZA 8329 & 01/29/19 & Jazz Air \ignore{Southwest} & GEG & OAK & 16:12 PST & 18:00 PST\\\hline
$f_8$ & SWA 1568 & 01/29/19 & Silkair & YYZ & YAM & 18:55 EST & \ignore{20:44}18:44 EST\\\hline
$f_9$ & SWA 1568\ignore{SWA 1959} & 01/29/19 & Southwest & LAX & OAK & 16:18 PST & 17:25 PST\\\hline
\end{tabular}
\end{minipage}
\endgroup
\end{table}

There are two potential answers to $q$, namely, `JZA 8329' and `SWA 1568', so we introduce their corresponding variables $p_1$ and $p_2$. Since the facts $f_1$ and $f_3$ form a key-equal group, we construct an $\alpha$-clause $(x_1 \lor x_3)$. Similarly, since the set $\{f_2, f_7\}$ of facts is a minimal witness to $q[\text{`JZA 8329'}]$, we construct the $\beta$-clause $(\neg x_2 \lor \neg x_7 \lor \neg p_1)$.
  By continuing this way, we obtain the following CNF-formula $\phi$:\\
$(x_1\lor x_3)\land x_2 \land x_4 \land x_5 \land x_6 \land x_7 \land (x_8 \lor x_9) \land (\neg x_2 \lor \neg x_7 \lor \neg p_1) \land (\neg x_3 \lor \neg x_9 \lor \neg p_2)$.\\
 Clauses $x_2$, $x_7$, and $(\neg x_2 \lor \neg x_7 \lor \neg p_1)$ force $p_1$ to take value 0 in each satisfying assignment of $\phi$,  because the facts $f_2$ and $f_7$ appear in every repair of $I$, thus making `JZA 8329' a consistent answer to $q$. In contrast, there is a satisfying assignment of $\phi$ in which $p_2$ is set to 1, which implies that  `SWA 1568' is not  a consistent answer to $q$.
\end{example}
\medskip
\noindent{\bf Optimizing the reductions}\label{optimizations}
In real-life applications, a large part of the inconsistent database is consistent. For a boolean query $q$, if a minimal witness to $q$ is present in the consistent part of the database instance, then we can immediately conclude that \certainty{$q$, $I$, $\Sigma$} is true. This can be checked with simple SQL queries that involve grouping on the key attributes of each relation. Similarly, for non-boolean queries, the consistent answers coming from the witnesses that belong to the consistent part of the database can be computed efficiently using SQL queries. All additional consistent answers can then be found using the preceding reduction. In this case, we need to introduce boolean variables corresponding to only those facts that contribute to the additional potential answers. This significantly reduces the size of the CNF-formulas produced by Reductions \ref{reduction1} or \ref{reduction2}. This optimization has been used earlier in \cite{Kolaitis13}, where  \cons{$q$, $I$, $\Sigma$} was reduced to an instance of binary integer programming.

\section{Consistent Query Answering Beyond Key Constraints}
In this section, we consider the  broader class of denial constraints and the more expressive class of unions of conjunctive queries. Note that computing the consistent answers of unions of conjunctive queries under denial constraints is still in coNP, but the consistent answers of a union $Q := q_1\cup \ldots \cup q_k$ of conjunctive queries $q_1,\ldots, q_k$ is not, in general, equal to the union of the consistent answers of $q_1,\ldots,q_k$.

We give a polynomial-time reduction from \cons{$Q$, $I$, $\Sigma$} to \unsat{}, where $\Sigma$ is a fixed finite set of denial constraints and $Q$ is a fixed union of non-boolean conjunctive queries. The potential answers to $Q$ are treated in the same way as the potential answers to the conjunctive query $q$ in Reduction \ref{reduction2}; to this effect, we  introduce a boolean variable for each potential answer.
The reduction we give here relies on the notions of \textit{minimal violations} and \textit{near-violations} to the set of denial constraints that we introduce next.

\begin{definition}\label{mv}
\textbf{Minimal violation.} \emph{Assume that $\Sigma$ is a set of denial constraints,  $I$ is an $\mathcal{R}$-instance, and  $S$ is a sub-instance of $I$. We say that $S$ is a \textit{minimal violation} to $\Sigma$, if $S \not\models \Sigma$  and for every set $S' \subset S$, we have that $S' \models \Sigma$.}
\end{definition}

\begin{definition}\label{nv}
\textbf{Near-violation.} \emph{Assume that $\Sigma$ is a set of denial constraints,  $I$ is an $\mathcal{R}$-instance,  $S$ is a sub-instance of $I$, and  $f$ is a fact of $I$. We say that $S$ is a \textit{near-violation} w.r.t.\ $\Sigma$ and $f$, if $S \models \Sigma$ and $S \cup \{f\}$ is a minimal violation to $\Sigma$. As a special case, if $\{f\}$ itself is a minimal violation to $\Sigma$, then we say that there is exactly one near-violation w.r.t. $f$, and it is the singleton  $\{f_{true}\}$, where $f_{true}$ is an auxiliary fact.}
\end{definition}

For a fixed finite set $\Sigma$ of denial constraints, the set of minimal violations to $\Sigma$ on a given database instance $I$ are computed as follows. The body of a denial constraint $d \in \Sigma$ is treated as a boolean conjunctive query $q_d$, possibly containing atomic formulas from $d$ that use built-in predicates such as $=$, $\neq$, $<$, $>$, $\leq$, and $\geq$, in addition to the relation symbols. The set of minimal witnesses to $q_d$ on $I$ is computed as described in Section \ref{sec:key-constraints}, which is also, precisely, the set of minimal violations to $d$. The union of the sets of minimal violations over all denial constraints in $\Sigma$ gives us the set of minimal violations to $\Sigma$. For each fact $f \in I$, the set of near-violations to $\Sigma$ w.r.t. $f$ can be obtained by removing $f$ from every minimal violation to $\Sigma$ that contains $f$.

Let $\mathcal{R}$ be a database schema, let $\Sigma$ be a fixed finite set of denial constraints on $\mathcal R$, and let $Q := q_1\cup \ldots \cup q_k$ be a union of conjunctive queries $q_1,\ldots,q_k$. Let $I$ be an $\mathcal{R}$-instance, and let $Q$ be the fixed union of $k$ non-boolean conjunctive queries $q_1, \ldots, q_k$.


\begin{reduction}\label{reduction3}\begin{sloppypar}
Given an $\mathcal R$-instance $I$, we construct a boolean formula $\phi'$ as follows.

\noindent Compute the following sets:


$\bullet$  $\mathcal{V}$: the set of minimal violations to $\Sigma$ on $I$.

$\bullet$  $\mathcal{N}^i$: the set of near-violations to $\Sigma$, on $I$, w.r.t. each fact $f_i \in I$.

$\bullet$ $\mathcal{A}$: the set of potential answers to $Q$ on $I$.

$\bullet$  $\mathcal{W}^l$: the set of all minimal witnesses to $Q[\vec{a}_l]$ on $I$, for each $\vec{a}_l \in \mathcal{A}$.


\noindent For each fact $f_i$ of $I$, introduce a boolean variable $x_i$, $1\leq i\leq n$. For the auxiliary fact $f_{true}$, introduce a constant  $x_{true} = true$.
 For each $N^i_j \in \mathcal{N}^i$, introduce a boolean variable $y^i_j$, and for each $\vec{a}_l \in \mathcal{A}$, introduce a boolean variable $p_l$.
\begin{enumerate}
\item For each $V_j \in \mathcal{V}$, construct a clause $\alpha_j = \underset{f_i \in V_j}{\lor} \neg x_i$.
\item For each $\vec{a}_l \in \mathcal{A}$ and for each $W^l_j \in \mathcal{W}^l$, construct a clause ${\beta^l_j = \bigg(\underset{f_i \in W^l_j}{\lor} \neg x_i\bigg) \lor \neg p_l}$.
\item For each $f_i \in I$, construct a clause $\gamma_i = x_i \lor \bigg(\underset{N^i_j \in \mathcal{N}^i}{\lor}y^i_j\bigg)$.
\item For each variable $y^i_j$, construct an expression $\theta^i_j = y^i_j \leftrightarrow \bigg(\underset{f_d \in N^i_j}{\land} x_d\bigg)$.
\item Construct the following boolean formula $\phi$:
$${\phi' = \bigg(\overset{|\mathcal{V}|}{\underset{i=1}{\land}}\alpha_i\bigg)\land\bigg(\overset{|\mathcal{A}|}{\underset{l=1}{\land}}\bigg(\overset{|\mathcal{W}^l|}{\underset{j=1}{\land}}\beta^l_j\bigg)\bigg) \land \bigg(\overset{|I|}{\underset{i=1}{\land}}\bigg(\Big(\overset{|\mathcal{N}^i|}{\underset{j=1}{\land}}\theta^i_j\Big)\land \gamma_i\bigg)\bigg)}$$
\end{enumerate}
\end{sloppypar}
\end{reduction}

\preservecounter{proposition}
\movetoappendix{\subsection*{Proof of  Proposition~\ref{prop3}}}
\copytoappendix{
\begin{proposition}\label{prop3} Let $\phi'$ be the boolean formula constructed using Reduction \ref{reduction3}.
\begin{itemize}
\item The formula $\phi'$ can be transformed to an equivalent CNF-formula $\phi$ whose size is polynomial in the size of $I$.
\item There exists a satisfying assignment to $\phi'$ in which a variable $p_l$ is set to 1 if and only if $\vec{a}_l \not\in \cons{Q, I, \Sigma}$.
\end{itemize}
\end{proposition}
}
\movetoappendix{
\begin{proof}
Let $n$ be the number of facts of $I$. Let $d_1$ be the smallest number such that there exists no denial constraint in $\Sigma$ whose number of database atoms is bigger than $d_1$. Also, let $d_2$ be the smallest number such that there exists no conjunctive query in $Q$ whose number of database atoms is bigger than $d_2$. Since $\Sigma$ and $Q$ are not part of the input to \cons{$Q$}, the quantities $d_1$ and $d_2$ are fixed constants.  We also have  that $|\mathcal{V}| \leq n^{d_1}$, $|\mathcal{N}^i| \leq n^{d_1}$ for $1 \leq i \leq n$, $|\mathcal{A}| \leq n^{d_2}$, and $|\mathcal{W}^l| \leq n^{d_2}$ for $1 \leq l \leq |\mathcal{A}|$. The number of $x$-, $y$-, and $p$-variables in $\phi'$ is therefore bounded by $n$, $n^{d_1 + 1}$, and $n^{d_2}$, respectively. The formula $\phi'$ contains as many $\alpha$-clauses as $|\mathcal{V}|$, and none of the $\alpha$-clause's length exceeds $n$. Similarly, there are at most $n^{d_2}$ $\beta$-clauses, and none of their lengths exceeds $d_2+1$. The number of $\gamma$-clauses is precisely $n$, and each $\gamma$-clause is at most $n^{d_1+1} + 1$ literals long. There are as many $\theta$-expressions as there are $y$-variables. Every $\theta$-expression is of the form $y \leftrightarrow (x_1 \land ... \land x_d)$, where $d$ is a constant obtained from the number of facts in the corresponding near-violation. Each $\theta$-expression can be equivalently written in a constant number of CNF-clauses as $((\neg y \lor x_1) \land ... \land (\neg y \lor x_d)) \land (\neg x_1 \lor ... \lor \neg x_d \lor y)$, in which the length each clause is constant. This makes it possible to transform $\phi'$ into an equivalent CNF-formula $\phi$, whose size is polynomial in the size of $I$.

 To prove the second part of the proposition, assume first that $\hat{a}$ is a satisfying assignment to the variables in $\phi'$ such that $\hat{a}(p_l) = 1$. Construct a database instance $R$ such that $f_i \in R$ if and only if $\hat{a}(x_i) = 1$. The $\alpha$-clauses make sure that no minimal violation to $\Sigma$ is present in $R$, meaning that $R$ is a consistent subset of $I$. The $\gamma$-clauses and the $\theta$-expressions encode the condition that, for every fact $f \in I$, either $f \in R$  or at least one near-violation w.r.t.\ $\Sigma$ and $f$ is in $R$. This condition makes sure that $R$ is indeed a repair of $I$. Since $\hat{a}(p_l) = 1$, the $\beta^l$-clauses ensure that at least one fact from each minimal witness to $Q[\vec{a}_l]$ is missing from $R$, meaning that $\vec{a}_l \not\in Q(R)$.

In the other direction, given a repair $R$ that falsifies $Q[\vec{a}_l]$, build an assignment $\hat{a}$ as follows. Set $\hat{a}(x_i) = 1$ if and only if $f_i \in r$. Set $\hat{a}(p_l) = 1$, and set $\hat{a}(p_{l'}) = 0$ for all $l' \neq l$. Since $R \models \Sigma$, no minimal violation to $\Sigma$ is a subset of $R$, meaning that $\hat{a}$ satisfies all $\alpha$-clauses in $\phi'$. Also, for every fact $f\in I$, it must be the case that either $f \in R$ or at least one near-violation w.r.t. $\Sigma$ and $f$ is in $R$ (otherwise $R$ would not have been a repair of $I$). Therefore, all $\gamma$-clauses and $\theta$-expressions are also satisfied by the assignment $\hat{a}$. Since $R \not\models Q[\vec{a}_l]$, at least one fact from each minimal witness to $Q[\vec{a}_l]$ must be missing from $R$, meaning that there is at least one variable $x_i$ in each $\beta^l$-clause such that $\hat{a}(x_i) = 0$. Hence, all $\beta^l$-clauses are satisfied by $\hat{a}$, even when $\hat{a}(p_l) = 1$. All other $\beta$-clauses are satisfied trivially, since $\hat{a}(p_{l'}) = 0$, for all $l' \neq l$.
\end{proof}
}
\begin{example}\label{example2}
Consider the database instance from Table \ref{flights}. In addition to the three key constraints from Example \ref{example1}, suppose the schema now has two additional integrity constraints: \begin{enumerate*}[label=(\alph*)]
\item if a flight departs from YYZ, then its airline must be Jazz Air;\label{egd} and
\item for Southwest airlines, if two tickets have the same code, then the ticket with an economy class must have lower fare than the one with the first class.\label{denial}
\end{enumerate*}
These can be expressed as the following denial constraints:
\begin{align*}
\text{\ref{egd} }\forall x, y, z, w, p, q\; \neg &(\textit{Flights}(x, y, z, \text{`YYZ'}, w, p, q) \land z \neq \text{`Jazz Air'})\\
\text{\ref{denial} }\forall x, y, z, w, p, q\; \neg &(\textit{Flights}(x, y, \text{`Southwest'}, z, w, p, q)\land\textit{Tickets}(r, x, \text{`First'}, t)\\
 &\land \textit{Tickets}(r', x, \text{`Economy'}, t') \land t \leq t')
\end{align*}

Let us say that we want to find the PNR numbers of the tickets booked with first class, or with Silkair airlines. This can be expressed as the union $Q:=q_1 \cup q_2$ of two unary conjunctive queries, where
\begin{align*}
q_1(x):=\; & \exists x, y, z\; \textit{Tickets}(x, y, \textit{`First'}, z)\\
q_2(x):=\; & \exists x, y, z, w, p, q, r, s, t\; \textit{Tickets}(x, y, z, w) \land \textit{Flights}(y, \textit{`Silkair'}, p, q, r, s, t)
\end{align*}

\begin{figure*}[!ht]
\caption{Minimal violations, minimal witnesses, and  near-violations in Example \ref{example2}.}
\begin{varwidth}[t]{.4\textwidth}
Minimal violations to $\Sigma$:
\begin{itemize}
\item $\{f_1, f_3\}$, $\{f_8\}$, $\{f_4, f_6, f_9\}$
\end{itemize}
Minimal witnesses to $Q$:
\begin{itemize}
\item $\{f_5\}$, $\{f_6\}$, $\{f_4, f_8\}$
\end{itemize}
\end{varwidth}
\hspace{4em}
\begin{varwidth}[t]{.6\textwidth}
Near-violations to $\Sigma$:\vspace{0.1in}\\
\begin{varwidth}[t]{.5\textwidth}
\begin{itemize}
\item $f_1$ : $\{f_3\}$
\item $f_3$ : $\{f_1\}$
\item $f_4$ : $\{f_6, f_9\}$
\item $f_6$ : $\{f_4, f_9\}$
\end{itemize}
\end{varwidth}
\begin{varwidth}[t]{.5\textwidth}
\begin{itemize}
\item $f_8$ : $\{f_{true}\}$
\item $f_9$ : $\{f_4, f_6\}$
\item $f_2, f_5, f_7$ : None
\end{itemize}
\end{varwidth}
\end{varwidth}
\label{fig1}
\end{figure*}
\begin{figure*}[!ht]
\centering
\caption{The $\alpha$-, $\beta$-, $\gamma$-clauses, and the $\theta$-expressions in Example \ref{example2}.}
\vspace{-0.2in}
\begin{align*}
\alpha\text{-clauses: } & (\neg x_1 \lor \neg x_3), (\neg x_8), (\neg x_4 \lor \neg x_6 \lor \neg x_9)\\
\beta\text{-clauses: } &(\neg x_5 \lor \neg p_1), (\neg x_6 \lor \neg p_2), (\neg x_4 \lor \neg x_8 \lor \neg p_3)\\
\gamma\text{-clauses: } & (x_1 \lor y^1_1), (x_2),(x_3 \lor y^3_1),(x_4 \lor y^4_1),(x_5), (x_6 \lor y^6_1), (x_7), (x_8 \lor y^8_1), (x_9 \lor y^9_1)\\
\theta\text{-expressions: } & (y^1_1 \leftrightarrow x_3), (y^3_1 \leftrightarrow x_1), (y^4_1 \leftrightarrow (x_6 \land x_9)), (y^6_1 \leftrightarrow (x_4 \land x_9)),  (y^8_1 \leftrightarrow x_{true}), \\
&(y^9_1 \leftrightarrow (x_4 \land x_6))
\end{align*}
\vspace{-0.2in}
\ignore{
\begin{align*}
\phi = & (\neg x_1 \lor \neg x_3) \land (\neg x_8) \land (\neg x_4 \lor \neg x_6 \lor \neg x_9) \land (\neg x_5 \lor \neg p_1) \land (\neg x_6 \lor \neg p_2)\\
 & \land (\neg x_4 \lor \neg x_8 \lor \neg p_3) \land (x_1 \lor x_3) \land (x_2) \land (x_4 \lor y^4_1) \land (x_5) \land (x_6 \lor y^6_1) \land (x_7)\\
 & \land (x_9 \lor y^9_1) \land (\neg y^4_1 \lor x_6) \land (\neg y^4_1 \lor x_9) \land (y^4_1 \lor x_6 \lor x_9) \land (\neg y^6_1 \lor x_4)\\
 & \land (\neg y^6_1 \lor x_9) \land (y^6_1 \lor x_4 \lor x_9) \land (\neg y^9_1 \lor x_4) \land (\neg y^9_1 \lor x_6) \land (y^9_1 \lor x_4 \lor x_6)
\end{align*}}
\label{fig2}
\end{figure*}

The  minimal witnesses to $Q$, the minimal violations to $\Sigma$, and the near-violations to $\Sigma$ w.r.t.\ each fact of the database are shown in Figure \ref{fig1}. With these, we construct the $\alpha$-, $\beta$-, $\gamma$-clauses, and the $\theta$-expressions of $\phi$, as shown in Figure \ref{fig2}.
 Even though, for  simplicity, it is not mentioned in Reduction \ref{reduction3},  we do the following optimization in practice: if $|N^i_j| = 1$, we do not introduce a variable $y^i_j$, but, we use the $x$-variable corresponding to the only fact in $N^i_j$. In each satisfying assignment to $\phi$, the variable $p_1$ must take the value 0. In contrast, this is not the case for $p_2$ and $p_3$. By Proposition \ref{prop3}, `KLF88V' is a consistent answer to $Q$, but `MJ9C8R' and `NJ5RT3' are not.
\end{example}

\section{Computing Consistent Answers via \wmaxsat{}}
By Proposition \ref{prop1}, the consistent answer to a boolean conjunctive query over a schema $\mathcal R$ with primary key constraints can be computed by solving the \unsat{} instance constructed in Reduction \ref{reduction1}. For non-boolean queries, however, in a CNF-formula $\phi$ constructed using Reduction \ref{reduction2} or \ref{reduction3}, one needs to identify each variable $p_l$ such that there exists at least one satisfying assignment to $\phi$ in which $p_l$ gets set to 1. By Proposition \ref{prop3}, the corresponding potential answers can then be discarded for being inconsistent. One way to do this is as follows. Add a clause $(p_1 \lor ... \lor p_{|A|})$ to $\phi$, and solve $\phi$ using a \sat{} solver. For each $p_l$ that gets set to 1 in the solution of $\phi$, remove  the literal $p_l$ from $\phi$  and then solve $\phi$ again. Repeat this process until $\phi$ is no longer satisfiable. At the end of this iterative process, the potential answers corresponding to the $p$-variables that still occur positively in $\phi$ are precisely the consistent answers to $Q$ on $I$. This approach, however, requires many \sat{} instances to be solved when the number of potential answers is large. For this reason, we developed and tested a different method  that uses solving \wmaxsat{} instances. The construction of these \wmaxsat{} instances is described  in Reduction \ref{reduction4}.
\begin{reduction}\label{reduction4}
Let the setup be the same as that of Reduction \ref{reduction2} (or Reduction \ref{reduction3}).
\begin{enumerate}
\item Construct a CNF-formula $\phi$ using Reduction \ref{reduction2} (or Reduction \ref{reduction3}).
\item Make all clauses in $\phi$ hard.
\item For each $\vec{a}_l \in \mathcal{A}$, construct a unit $\epsilon$-clause $\epsilon_l = (p_l)$.
\item Make all $\epsilon$-clauses soft, and of equal weights.
\item Construct the WCNF-formula $\psi = \phi \land \bigg(\overset{|\mathcal{A}|}{\underset{l=1}{\land}}\epsilon_l\bigg)$.
\end{enumerate}
\end{reduction}
\vspace{-0.13in}
\begin{algorithm}[!ht]
\caption{Eliminating Inconsistent Potential Answers}\label{alg:eliminate}
\begin{algorithmic}[1]
\Procedure{EliminateWithMaxSAT}{$\psi, \mathcal{A}$}
\State \textbf{let} \textsc{Ans} $=$ \textbf{bool array}$[|\mathcal{A}|]$
\For {$l = 1$ to $|\mathcal{A}|$}
\State \textsc{Ans}[$l$] $\gets$ true
\EndFor\State \textbf{let bool} \textit{moreAnswers} $\gets$ true
\While{\textit{moreAnswers}}
\State \textit{moreAnswers} $\gets$ false
\State \textbf{let} \textit{opt} $\gets$ \maxsat{$(\psi)$}\Comment{Use \wmaxsat{} solver}
\For {$l=1$ to $|\mathcal{A}|$}
\If{\textit{opt}[$p_l$] = 1}
\State \textit{moreAnswers} $\gets$ true
\State \textsc{Ans}[$l$] $\gets$ false
\State Remove the unit clause $(p_l)$ from $\psi$
\State Remove all clauses containing the literal $\neg p_l$ from $\psi$
\State Add a new unit hard clause $(\neg p_l)$ to $\psi$
\EndIf
\EndFor
\EndWhile
\State \textbf{return} \textsc{Ans}
\EndProcedure
\end{algorithmic}
\end{algorithm}

 The preceding Algorithm \ref{alg:eliminate}  computes the consistent answers by iteratively solving \wmaxsat{} instances. It takes as inputs the instance $\psi$ constructed using Reduction \ref{reduction4}  and the set $\mathcal A$ of potential answers. The idea is to eliminate, in each iteration, as many inconsistent answers from $\mathcal{A}$ as possible  by solving $\psi$. After each iteration, $\psi$ is modified in such a way that additional inconsistent answers, if any,  can be eliminated in  subsequent iterations.
 In  Section \ref{experiments}, we carried out experiments in which it turned out that  the number of iterations taken by Algorithm \ref{alg:eliminate} is less than 4, even when there is a large number of potential answers.


\movetoappendix{\subsection*{Proof of Proposition~\ref{prop4}}}
\movetoappendix{
We first state and prove Lemma \ref{lemma1}, that reasons about the satisfying assignments to the WCNF-formula $\psi$, constructed using Reduction \ref{reduction4}. This lemma is used in proving Proposition \ref{prop4}.
\begin{lemma}\label{lemma1}
Let $\phi'$ be the CNF-formula constructed in Step 1 of Reduction \ref{reduction4}, and $\psi_i$ be the \wmaxsat{} instance at the beginning of $i^{th}$ iteration of Algorithm \ref{alg:eliminate}. For all $i$, every optimal solution of $\psi_i$ satisfies all clauses in $\phi'$.
\end{lemma}
\begin{proof}
We prove Lemma \ref{lemma1} by induction on $i$. The CNF-formula $\phi'$ constructed in Step 1 of Reduction \ref{reduction4} can always be satisfied by setting all $x$-variables to 1, and all $p$-variables to 0. The clauses in $\phi$ being hard, $\epsilon$ being soft ensures that every optimal solution of $\psi_0$ satisfies all clauses in $\phi$. Assume that for some $i \geq 0$, every optimal solution to $\psi_i$ satisfies all clauses in $\phi$. At the end of iteration $i+1$, if \textit{moreAnswers} is true, the formula $\psi_{i+1}$ is constructed from $\psi_i$ by adding to $\psi_i$ the unit hard clauses $(\neg p_l)$. This forces every optimal solution of $\psi_{i+1}$ to satisfy all of these added clauses. Since no $p_l$ variable occurs positively in $\phi$, we have that, for all $i$, every optimal solution to $\psi_{i+1}$ still satisfies $\phi$.
\end{proof}
}
\copytoappendix{
\begin{proposition}\label{prop4} Algorithm \ref{alg:eliminate} returns an array \textsc{Ans} such that $\vec{a}_l \in \cons{Q, I}$ if and only if the entry \textsc{Ans[$l$]} is true.
\end{proposition}
}
\movetoappendix{
\begin{proof}
In one direction, for every $l$, if $\vec{a}_l \in \cons{Q, I}$, then, by Proposition \ref{prop3}, the variable $p_l$ takes value 0 in every  assignment that satisfies $\phi$. By Lemma \ref{lemma1}, for every $i$, the optimal solution of $\psi_i$ also assigns value 0 to the variable $p_l$. As a result, Line 12 never gets executed, and the entry \textsc{Ans}[$l$] remains true.

For the other direction, we first prove that Algorithm \ref{alg:eliminate} always terminates. Observe that at the end of the $i^{th}$ iteration, for every $l$, a unit clause $(p_l)$ is present in $\psi_i$ if and only if \textsc{Ans}[$l$] is true. Hence, at the end of $i^{th}$ iteration, if \textit{moreAnswers} is true, then the optimal solution to $\psi_i$ must have assigned value 1 to at least one variable $p_l$ such that \textsc{Ans}[$l$] was previously true. Therefore, at the end of $i^{th}$ iteration, at least $i$ entries in \textsc{Ans} are false. It follows  that the algorithm terminates after at most $|\mathcal {A}|$ iterations.

Since no clause in $\phi$ contains a positive literal $p_l$, the addition of a unit hard clause $(\neg p_l)$ to $\psi_i$ does not suppress any satisfying assignments to $\phi$ while finding the optimal solution to $\psi_{i+1}$. Therefore, in every iteration, the optimal solution of $\psi$ guarantees to satisfy the maximum number of $p_l$ variables for which the unit clause $(p_l)$ is still in $\psi$. As a result, Algorithm \ref{alg:eliminate} does not terminate until it marks the entries \textsc{Ans}[$l$] false, for all $l$, for which there exists a satisfying assignment to $\phi$ in which $p_l$ gets assigned to 1. In other words, by Proposition \ref{prop3}, for every inconsistent answer $\vec{a}_l$, the entry \textsc{Ans} gets marked as false.
\end{proof}
}
\section{Preliminary Experimental Results}\label{experiments}
\looseness = -1 We evaluated the performance of CAvSAT using two different scenarios. First, we experimented with large synthetically generated databases having primary key constraints. We implemented Reduction \ref{reduction2} without the optimization mentioned in Section \ref{optimizations}. We also implement Reduction \ref{reduction4} and Algorithm \ref{alg:eliminate}. We found out that for seven non-boolean FO-rewritable queries that were also used in \cite{Kolaitis13}, CAvSAT significantly outperformed the database evaluation of the FO-rewritings obtained using the algorithm from \cite{Koutris17}. We also implemented Reduction \ref{reduction2} with the optimization, and evaluated its performance on fourteen additional conjunctive queries whose consistent answers are coNP-complete or are in P but are not FO-rewitable. In the second scenario, we evaluated the performance of CAvSAT using Reduction \ref{reduction3} on a real-world database with functional dependencies.
The definitions of the queries used in the experiments, as well as the FO-rewritings of the first seven queries, can be found in the Appendix.

\medskip
\noindent{\bf Experimental Setup}
All experiments were carried out on a machine running on Intel Core i7 2.7 GHz, 64 bit Ubuntu 16.04, with 8GB of RAM. We used PostgreSQL 10.1 as an underlying DBMS, and MaxHS v3.0 solver \cite{Davies2011} for solving the \wmaxsat{} instances. Our system is implemented in Java 9.04.
\subsection{Synthetic Data Generation}
The synthetic data were generated in two phases: (a) generation of consistent data;   (b) injection of inconsistency into  consistent data. The parameters used to generate the data were the number of tuples per relation (\textit{rSize}), degree of inconsistency (\textit{inDeg}), and the size of each key-equal group (\textit{kSize}). \\
\textbf{Generating consistent data} Each relation in the consistent database was generated with the same number of tuples, so that injecting inconsistency with specified \textit{kSize} and \textit{inDeg} will make the total number of tuples in the relation equal to \textit{rSize}. For each query used in the experiment, the data was generated so the evaluation of the query on the consistent database results in a relation that has the size 15\% to 20\% of \textit{rSize}. The values of the third attribute in all of the ternary relations, were chosen from a uniform distribution in the range [1, \textit{rSize}/10]. This was done to simulate a reasonably large number of potential answers. The remaining attributes take values from randomly generated alphanumeric strings of length 10.\\
\textbf{Injecting inconsistency} In each relation, the inconsistency was injected by inserting new tuples to the consistent data, that share the values of the key attributes with some already existing tuples from the consistent data. The parameter \textit{inDeg} denotes in percentage the number of tuples per relation, that participate in a key violation. We conducted experiments with the varying values for \textit{inDeg}, ranging from 5\% to 15\%. The values of \textit{kSize} were uniformly distributed between 2 to 5. The non-key attributes of the newly injected tuples were uniform random alphanumeric strings of length 10.
\subsection{Experimental Results}
\noindent{\bf CAvSAT on FO-rewritable Queries }
In this set of experiments, we compare the performance of CAvSAT against the FO-rewritings of seven queries over the database with primary key constraints. For queries $q_1, \dots, q_7$, we computed the FO-rewritings using the algorithm of Koutris and Wijsen in \cite{Koutris17}. We refer to these FO-rewritings as \textit{KW-FO-rewritings}. Since the queries $q_1, \dots, q_4$ happen to be in the class $C_\textit{forest}$, we computed additional FO-rewritings for them using the algorithm implemented in the consistent query answering system ConQuer  \cite{Fuxman05};  we refer to these rewritings as \textit{ConQuer-FO-rewritings}. 
Each FO-rewriting was translated into SQL, and fed to PostgreSQL.

Table \ref{formula-size} shows the size of the WCNF-formulas produced by Reduction \ref{reduction4} without optimization (where Reduction \ref{reduction2} is used inside Reduction \ref{reduction4}), on these queries over the databases having one million tuples per relation. Figure \ref{unopt} shows the evaluation time of CAvSAT with these formulas using MaxHS v3.0 solver \cite{Davies2011}. The letter E denotes the time required for encoding the problem into a \wmaxsat{} instance, and the letter S denotes the time taken by Algorithm \ref{alg:eliminate}. The percentage adjacent to the letters S and E denotes the degree of inconsistency. Figure \ref{opt-vs-unopt} (left) shows the significant gain in performance due to the optimization, for databases of size one million tuples per relation and with 10\% inconsistency. This is not surprising; since 90\% of the data were consistent, it is expected that most of the consistent answers lie in the consistent part of the database. Table \ref{formula-size-optimized} shows the size of the WCNF-formulas produced by Reduction \ref{reduction4}, with the optimization in place.

Figure \ref{fo-vs-sat} (right) shows that for the queries $q_1, \dots, q_4$ in the class $C_\textit{forest}$, the performance of CAvSAT is slightly worse, but comparable, to their ConQuer-FO-rewritings. For all seven queries $q_1, \dots, q_7$, however, CAvSAT significantly outperformed their KW-FO-rewritings, as PostgreSQL hit the two hours timeout while evaluating each KW-FO-rewriting. In fact, this timeout was hit by all seven queries even for databases of size as small as 100K tuples per relation.
  For $q_1, \dots, q_7$, the average number of iterations taken by Algorithm \ref{alg:eliminate} to eliminate all inconsistent potential answers was 2.85.

\begin{table}[ht]
\centering
  \begin{varwidth}[t]{0.4\linewidth}
  \caption{Size of CNF-formula}
  \begin{center}
  \renewcommand{\arraystretch}{1.1}
\setlength{\tabcolsep}{3pt}
    \begin{tabular}[h]{|c|c|c|}\hline
    \centering
Query & Variables & Clauses
\\\hline\hline
$q_1$ & 2.08M & 2.18M \\\hline
$q_2$ & 2.07M & 2.12M\\\hline
$q_3$ & 3.15M & 3.07M\\\hline
$q_4$ & 3.23M & 3.07M\\\hline
$q_5$ & 2.07M & 2.12M \\\hline
$q_6$ & 3.3M & 3.06M\\\hline
$q_7$ & 3.25M & 3.06M \\\hline
\end{tabular}\label{formula-size}
\end{center}
  \end{varwidth}
  \begin{minipage}[t]{0.59\linewidth}

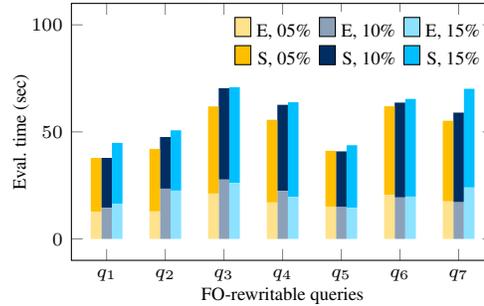
\captionof{figure}{Evaluation time of CAvSAT without optimization, for 1M tuples/relation.}
    \begin{tikzpicture}
\begin{axis}[ybar stacked,
		width=\textwidth,height=5cm,
		xlabel= FO-rewritable queries,
		label style={font=\scriptsize},
		xtick=data,
		xticklabels={$q_1$,$q_2$,$q_3$,$q_4$,$q_5$,$q_6$,$q_7$},
		legend style={at={(0.8,0.97)},font=\scriptsize,draw=none},
		legend cell align = {right},
		bar width=4pt,
		ymin=0, ymax=100,
		enlargelimits=0.1,
		y label style={at={(0.1,0.5)}},
		x label style={at={(0.5,0.08)}},
		xmin = 1,xmax=7,
		ylabel=Eval. time (sec)]
	\addplot[fill=coolblack!45!white,draw=none] coordinates {
		(1,14.532)
		(2,23.493)
		(3,27.736)
		(4,22.455)
		(5,14.87)
		(6,19.376)
		(7,17.129)
	};\addlegendentry{E, 10\%}
	\addplot[fill=coolblack, draw=none] coordinates {
		(1,23.328)
		(2,24.103)
		(3,42.657)
		(4,40.179)
		(5,26.027)
		(6,44.297)
		(7,41.870)
	};\addlegendentry{S, 10\%}
	\end{axis}
	
	\begin{axis}[ybar stacked, bar shift=-4pt,
		width=\textwidth,height=5cm,
		legend style={at={(0.6,0.97)},font=\scriptsize,draw=none,fill=none},
		legend cell align = {right},
		xticklabels=none,axis lines=none, xtick=\empty, ytick=\empty,
		bar width=4pt,
		ymin=0, ymax=100,
		enlargelimits=0.1,
		xmin = 1,xmax=7]
	\addplot[fill=amber!45!white,draw=none] coordinates {
		(1,12.598)
		(2,12.844)
		(3,21.19)
		(4,17.03)
		(5,15.13)
		(6,20.524)
		(7,17.527)
	};\addlegendentry{E, 05\%}
	\addplot[fill=amber, draw=none] coordinates {
		(1,25.237)
		(2,29.21)
		(3,40.704)
		(4,38.609)
		(5,26.015)
		(6,41.44)
		(7,37.66)
	};\addlegendentry{S, 05\%}
	\end{axis}
	
	\begin{axis}[ybar stacked, bar shift=4pt,
		width=\textwidth,height=5cm,
		legend style={at={(1,0.97)},font=\scriptsize,draw=none,fill=none},
		legend cell align = {right},
		xticklabels=none,axis lines=none, xtick=\empty, ytick=\empty,
		bar width=4pt,
		ymin=0, ymax=100,
		enlargelimits=0.1,
		xmin = 1,xmax=7]
	\addplot[fill=capri!45!white,draw=none] coordinates {
		(1,16.45)
		(2,22.583)
		(3,26.13)
		(4,19.581)
		(5,14.59)
		(6,19.777)
		(7,24.022)
	};\addlegendentry{E, 15\%}
	\addplot[fill=capri, draw=none] coordinates {
		(1,28.412)
		(2,28.15)
		(3,44.713)
		(4,44.265)
		(5,29.185)
		(6,45.624)
		(7,46.106)
	};\addlegendentry{S, 15\%}
	\end{axis}
\end{tikzpicture}
    \label{unopt}
  \end{minipage}
\end{table}
\vspace{-4em}
\begin{figure*}[h]
\caption{Evaluation time of CAvSAT with and without optimization (left). Evaluation time of CAvSAT with optimization, in comparison with the KW-FO-rewriting and the ConQuer-FO-rewriting, for 1M tuples/relation with 10\% inconsistency (right).}
\begin{varwidth}[t]{.55\textwidth}
\centering
\begin{tikzpicture}
\begin{axis}[ybar stacked,bar shift =-2.5pt,
		width=\textwidth,height=5cm,
		xlabel= FO-rewritable queries,
		label style={font=\scriptsize},
		xtick=data,
		xticklabels={$q_1$,$q_2$,$q_3$,$q_4$,$q_5$,$q_6$,$q_7$},
		legend style={at={(0.8,0.97)},font=\scriptsize,draw=none,fill=none},
		legend cell align = {left},
		bar width=4pt,
		ymin=0, ymax=80,
		enlargelimits=0.1,
		y label style={at={(0.1,0.5)}},
		x label style={at={(0.5,0.08)}},
		xmin = 1,xmax=7,
		ylabel=Eval. time (sec)]
	\addplot[fill=asparagus!50!white,draw=none] coordinates {
		(1,14.532)
		(2,23.493)
		(3,27.736)
		(4,22.455)
		(5,14.87)
		(6,19.376)
		(7,17.129)
	};\addlegendentry{E, Unopt}
	\addplot[fill=asparagus, draw=none] coordinates {
		(1,23.328)
		(2,24.103)
		(3,42.657)
		(4,40.179)
		(5,26.027)
		(6,44.297)
		(7,41.870)
	};\addlegendentry{S, Unopt}
	\end{axis}

\begin{axis}[ybar stacked, bar shift = 2.5pt,
		width=\textwidth,height=5cm,
		xlabel= \empty,ylabel=\empty,yticklabels=\empty,
		label style={font=\scriptsize},
		xtick=data,
		xticklabels=\empty,
		legend style={at={(1,0.97)},font=\scriptsize,draw=none,fill=none},
		legend cell align = {left},
		bar width=4pt,
		ymin=0, ymax=80,
		enlargelimits=0.1,
		xmin = 1,xmax=7,
		ylabel=\empty]
	\addplot[fill=coolblack!30!white,draw=none] coordinates {
		(1,7.265)
		(2,8.205)
		(3,8.855)
		(4,9.644)
		(5,7.707)
		(6,4.46)
		(7,5.27)
	};\addlegendentry{E, Opt}
	\addplot[fill=coolblack, draw=none] coordinates {
		(1,0.478)
		(2,1.013)
		(3,0.591)
		(4,1.427)
		(5,0.306)
		(6,0.499)
		(7,0.456)
	};\addlegendentry{S, Opt}
	\end{axis}
\end{tikzpicture}
\label{opt-vs-unopt}
\end{varwidth}
\hspace{1em}
\begin{varwidth}[t]{.45\textwidth}
\begin{tikzpicture}
	\begin{axis}[
		width=\textwidth,height=5cm,
		xlabel=FO-rewritable queries,
		xtick=data, ytick={0,5,10,15,20},scaled y ticks = false,
		ymin=5, ymax=30,
		enlargelimits=0.15,
		legend style={at={(0.85,0.95)},font=\scriptsize,draw=none,fill=none},
		legend cell align = {left},
		x label style={at={(0.5,0.08)}},
		label style={font=\scriptsize},
		xticklabels={$q_1$,$q_2$,$q_3$,$q_4$,$q_5$,$q_6$,$q_7$},
		yticklabels={0,5,10,15,2hr+},y label style={at={(0.12,0.5)}},
		xmin = 1,xmax=7,every axis plot/.append style={thick},
		ylabel=Eval. time (sec)]
	\addplot[color=red,mark=square] coordinates {
		(1,20)
		(2,20)
		(3,20)
		(4,20)
		(5,20)
		(6,20)
		(7,20)
	};
	\addlegendentry{KW-FO-rewriting}
	\addplot[color=coolblack,mark=o] coordinates {
	(1,7.743)
		(2,9.218)
		(3,9.476)
		(4,11.07)
		(5,8.013)
		(6,4.959)
		(7,5.726)
	};
	\addlegendentry{CavSAT}
	\addplot[color=capri,mark=triangle] coordinates {
		(1,6.553)
		(2,6.826)
		(3,9.107)
		(4,9.528)
	};
	\addlegendentry{ConQuer-FO-rewriting}
	\end{axis}
\end{tikzpicture}
\end{varwidth}
\label{fo-vs-sat}
\end{figure*}
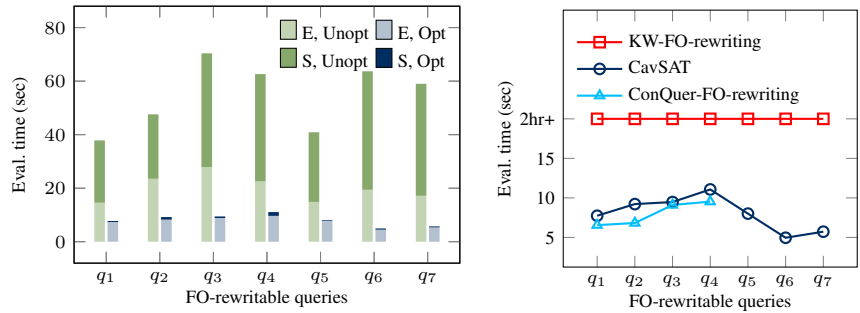
\medskip
\noindent {\bf CAvSAT on Harder Queries}
In this set of experiments, we considered fourteen additional non-boolean conjunctive queries whose consistent answers are coNP-complete or in P but not FO-rewritable (Figure \ref{fo-queries} in the Appendix). Figure \ref{hard-queries} shows that the time required for the optimizing and then constructing the \wmaxsat{} instance using Reduction \ref{reduction4}, dominates over the time taken by Algorithm \ref{alg:eliminate}. The solver takes comparatively more time for the queries that have more free variables or more atoms. Table \ref{formula-size-optimized} shows the size of the CNF-formulas constructed by Reduction \ref{reduction4} (where Reduction \ref{reduction2} is used inside Reduction \ref{reduction4}) in this experiment. The average number of iterations taken by Algorithm \ref{alg:eliminate} to eliminate all inconsistent potential answers to a query was 3.2.
\vspace{-0.4in}
\begin{table}[ht]
\caption{The size of the CNF-formulas with optimization, for 1M tuples per relation.}
\renewcommand{\arraystretch}{1.1}
\setlength{\tabcolsep}{3pt}
\centering
\begin{tabular}[h]{|c|c|c||c|c|c||c|c|c|}\hline
Query & Variables & Clauses & Query & Variables & Clauses & Query & Variables & Clauses
\\\hline\hline
$q_1$ & 16.5K & 20.9K & $q_8$ & 16.6K & 16.8K & $q_{15}$ & 14.9K & 15K\\\hline
$q_2$ & 68.6K& 76.0K & $q_9$ & 58K & 57.7K& $q_{16}$ & 58.8K & 58.4K \\\hline
$q_3$ &31.9K &36.8K & $q_{10}$ & 31.3K & 36.6K & $q_{17}$ & 40.1K & 41.4K\\\hline
$q_4$ & 117.2K& 123.7K& $q_{11}$ & 105K & 118.1K & $q_{18}$ & 107.5K & 121.4K\\\hline
$q_5$ & 16.3K&20.6K& $q_{12}$ & 116.8K & 123.4K & $q_{19}$ & 114.4K & 120.7K \\\hline
$q_6$ & 32.8K & 33.2K & $q_{13}$ & 63.2K & 65.7K & $q_{20}$ & 53.4K & 63.7K\\\hline
$q_7$ &32.5K & 33.8K & $q_{14}$ & 53.9K & 59.2K & $q_{21}$ & 170K & 199K\\\hline
\end{tabular}\label{formula-size-optimized}
\end{table}
\vspace{-0.3in}
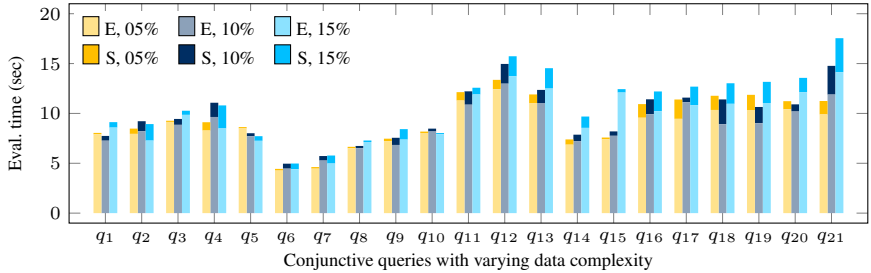
\begin{figure*}[h]
\caption{Evaluation time of CAvSAT for conjunctive queries with varying data complexity, with optimization, over the databases of size 1M tuples/relation.}
\begin{tikzpicture}
\begin{axis}[ybar stacked,
		width=\textwidth,height=4.5cm,
		xlabel= Conjunctive queries with varying data complexity,
		label style={font=\scriptsize},
		xtick=data,
		xticklabels={$q_1$,$q_2$,$q_3$,$q_4$,$q_5$,$q_6$,$q_7$,$q_8$,$q_9$,$q_{10}$,$q_{11}$,$q_{12}$,$q_{13}$,$q_{14}$,$q_{15}$,$q_{16}$,$q_{17}$,$q_{18}$,$q_{19}$,$q_{20}$,$q_{21}$},
		legend style={at={(0.25,0.97)},font=\scriptsize,draw=none},
		legend cell align = {right},
		bar width=3pt,
		ymin=0, ymax=20,
		enlargelimits=0.05,
		y label style={at={(0.05,0.5)}},
		x label style={at={(0.5,0.08)}},
		xmin = 1,xmax=21,
		ylabel=Eval. time (sec)]
	\addplot[fill=coolblack!45!white,draw=none] coordinates {
		(1,7.265)
		(2,8.205)
		(3,8.855)
		(4,9.644)
		(5,7.707)
		(6,4.46)
		(7,5.27)
		(8,6.504)
		(9,6.838)
		(10,8.187)
		(11,10.875)
		(12,12.983)
		(13,11.045)
		(14, 7.226)
		(15,7.756)
		(16,9.925)
		(17, 11.13)
		(18,8.913)
		(19,9.016)
		(20,10.19)
		(21, 11.889)
	};\addlegendentry{E, 10\%}
	\addplot[fill=coolblack, draw=none] coordinates {
		(1,0.478)
		(2,1.013)
		(3,0.591)
		(4,1.427)
		(5,0.306)
		(6,0.499)
		(7,0.456)
		(8,0.233)
		(9,0.718)
		(10,0.287)
		(11,1.34)
		(12,1.98)
		(13,1.313)
		(14,0.642)
		(15,0.440)
		(16,1.49)
		(17, 0.455)
		(18,2.49)
		(19,1.63)
		(20,0.719)
		(21, 2.88)
	};\addlegendentry{S, 10\%}
	\end{axis}
	
	\begin{axis}[ybar stacked, bar shift=-3pt,
		width=\textwidth,height=4.5cm,
		legend style={at={(0.13,0.97)},font=\scriptsize,draw=none},
		legend cell align = {right},
		xticklabels=none,
		bar width=3pt,axis lines=none, xtick=\empty, ytick=\empty,
		ymin=0, ymax=20,
		enlargelimits=0.05,
		xmin = 1,xmax=21]
	\addplot[fill=amber!45!white,draw=none] coordinates {
		(1,7.904)
		(2,7.948)
		(3,9.146)
		(4,8.28)
		(5,8.496)
		(6,4.267)
		(7,4.461)
		(8,6.508)
		(9,7.223)
		(10,8.037)
		(11, 11.289)
		(12,12.43)
		(13,11.007)
		(14, 6.865)
		(15,7.397)
		(16,9.567)
		(17, 9.456)
		(18,10.347)
		(19,10.343)
		(20,10.413)
		(21, 9.896)
	};\addlegendentry{E, 05\%}
	\addplot[fill=amber, draw=none] coordinates {
		(1,0.142)
		(2,0.528)
		(3,0.12)
		(4,0.836)
		(5,0.133)
		(6,0.170)
		(7,0.158)
		(8,0.142)
		(9,0.243)
		(10,0.141)
		(11,0.844)
		(12,0.936)
		(13, 0.897)
		(14,0.527)
		(15,0.184)
		(16,1.367)
		(17,1.94)
		(18,1.422)
		(19,1.521)
		(20,0.815)
		(21,1.349)
	};\addlegendentry{S, 05\%}
	\end{axis}
	
	\begin{axis}[ybar stacked, bar shift=3pt,
		width=\textwidth,height=4.5cm,
		legend style={at={(0.37,0.97)},font=\scriptsize,draw=none},
		legend cell align = {right},
		xticklabels=none,axis lines=none, xtick=\empty, ytick=\empty,
		bar width=3pt,
		ymin=0, ymax=20,
		enlargelimits=0.05,
		xmin = 1,xmax=21]
	\addplot[fill=capri!45!white,draw=none] coordinates {
		(1,8.578)
		(2,7.273)
		(3,9.855)
		(4,8.491)
		(5,7.252)
		(6,4.403)
		(7,4.973)
		(8,7.132)
		(9,7.373)
		(10,7.91)
		(11,11.899)
		(12, 13.72)
		(13, 12.493)
		(14, 8.55)
		(15,12.093)
		(16,10.221)
		(17, 10.836)
		(18,10.964)
		(19, 11.042)
		(20, 12.115)
		(21, 14.143)
	};\addlegendentry{E, 15\%}
	\addplot[fill=capri, draw=none] coordinates {
		(1,0.547)
		(2,1.657)
		(3,0.414)
		(4,2.312)
		(5,0.4640)
		(6,0.571)
		(7,0.804)
		(8,0.152)
		(9,1.051)
		(10,0.12)
		(11,0.676)
		(12,2.013)
		(13, 2.043)
		(14,1.137)
		(15,0.342)
		(16,1.981)
		(17,1.85)
		(18,2.06)
		(19,2.124)
		(20,1.45)
		(21,3.401)
	};\addlegendentry{S, 15\%}
	\end{axis}
\end{tikzpicture}
\label{hard-queries}
\end{figure*}

\movetoappendix{
\begin{figure*}
\caption{Queries used in the experiments with synthetic data.}
\begin{itemize}
\item[] \textbf{FO-rewritable consistent answers:}
\item[] $q_1(z):=\; \exists x, y, v, w\; (R_1(\underline{x},y,z) \land R_2(\underline{y},v,w))$
\item[] $q_2(z,w):=\; \exists x, y, v\; (R_1(\underline{x},y,z) \land R_2(\underline{y}, v,w))$
\item[] $q_3(z):=\; \exists x, y, v, u, d\; (R_1(\underline{x},y,z) \land R_3(\underline{y},v) \land R_2(\underline{v},u,d))$
\item[] $q_4(z,d):=\; \exists x, y, v, u\; (R_1(\underline{x},y,z) \land R_3(\underline{y},v) \land R_2(\underline{v},u,d))$
\item[] $q_5(z):=\; \exists x, y, v, w\; (R_1(\underline{x},y,z) \land R_4(\underline{y,v},w))$
\item[] $q_6(z):=\; \exists x, y, x', w, d\; (R_1(\underline{x},y,z) \land R_2(\underline{x'},y,w) \land R_5(\underline{x},y,d))$
\item[] $q_7(z):=\; \exists x, y, w, d\; (R_1(\underline{x},y,z) \land R_2(\underline{y},x, w) \land R_5(\underline{x},y,d))$\\
\item[] \textbf{In P, but not FO-rewritable, consistent answers:}
\item[] $q_8(z,w):=\; \exists x, y\; (R_1(\underline{x},y,z) \land R_2(\underline{y},x,w))$
\item[] $q_9(z):=\; \exists x, y, w, u, d\; (R_1(\underline{x},y,z) \land R_2(\underline{y},x,w) \land R_4(\underline{y},u,d))$
\item[] $q_{10}(z,w,d):=\; \exists x, y, u\; (R_1(\underline{x},y,z) \land R_2(\underline{y},x,w) \land R_4(\underline{y},u,d))$
\item[] $q_{11}(z):=\; \exists x, y, w\; (R_1(\underline{x},y,z) \land R_2(\underline{y},x,w))$
\item[] $q_{12}(v, d):=\; \exists x, y, z, u\; (R_3(\underline{x},y) \land R_6(\underline{y},z) \land R_1(\underline{z},x, d) \land R_4(\underline{x, u},v))$
\item[] $q_{13}(v):=\; \exists x, y, z, u\; (R_3(\underline{x},y) \land R_6(\underline{y},z) \land R_7(\underline{z},x) \land R_4(\underline{x, u},v))$
\item[] $q_{14}(d):=\; \exists x, y, z, u\; (R_3(\underline{x},y) \land R_6(\underline{y},z) \land R_1(\underline{z},x,d) \land R_7(\underline{x},u))$
\\
\item[] \textbf{coNP-complete consistent answers:}
\item[] $q_{15}(z):=\; \exists x, y, x', w\; (R_1(\underline{x},y,z) \land R_2(\underline{x'},y,w))$
\item[] $q_{16} (z,w):=\; \exists x, y, x'\; (R_1(\underline{x},y,z) \land R_2(\underline{x'},y,w))$
\item[] $q_{17}(z):=\; \exists x, y, x', w, u, d\; (R_1(\underline{x},y,z) \land R_2(\underline{x'},y,w) \land R_4(\underline{y},u,d))$
\item[] $q_{18}(z,w):=\; \exists x, y, x', u, d\; (R_1(\underline{x},y,z) \land R_2(\underline{x'},y,w) \land R_4(\underline{y},u,d))$
\item[] $q_{19}(z,w,d):=\; \exists x, y, x', u\; (R_1(\underline{x},y,z) \land R_2(\underline{x'},y,w) \land R_4(\underline{y},u,d))$
\item[] $q_{20}(z):=\; \exists x, y, x', w, u, d, v\; (R_1(\underline{x},y,z) \land R_2(\underline{x'},y,w) \land R_4(\underline{y},u,d) \land R_3(\underline{u},v))$
\item[] $q_{21}(z,w):=\; \exists x, y, x', u, d, v\; (R_1(\underline{x},y,z) \land R_2(\underline{x'},y,w) \land R_4(\underline{y},u,d) \land R_3(\underline{u},v))$
\end{itemize}\label{fo-queries}
\end{figure*}

\begin{figure*}[!ht]
\caption{The KW-FO-rewritings of the consistent answers to queries $q_1$ to $q_7$.}
\begin{align*}
q_1(f^{(1)}):=\; &\exists s_1 \in R_1 (\forall r_1 \in R_1(\exists s_2 \in R_2(\forall r_2 \in R_2
\\ &(s_1.1 \neq r_1.1 \lor s_2.1 \neq r_2.1 \lor (r_1.2 = r_2.1 \land r_1.3 = f.1)))))\\
q_2(f^{(2)}):=\; &\exists s_1 \in R_1(\forall r_1 \in R_1(\exists s_2 \in R_2(\forall r_2 \in R_2
\\ &(s_1.1 \neq r_1.1 \lor s_2.1 \neq r_2.1 \lor (r_2.3 = f.1 \land r_1.2 = r_2.1 \land r_1.3 = f.2)))))\\
q_3(f^{(1)}):=\; &\exists s_1 \in R_1(\forall r_1 \in R_1(\exists s_2 \in R_3(\forall r_2 \in R_3(\exists s_3 \in R_2(\forall r_3 \in R_2
\\&((s_1.1 \neq r_1.1 \lor s_2.1 \neq r_2.1 \lor s_3.1 \neq r_3.1 \lor (r_2.2 = r_3.1 \land r_1.2 = r_2.1 \land r_1.3 = f.1))))))))\\
q_4(f^{(2)}):=\; &\exists s_1 \in R_1(\forall r_1 \in R_1(\exists s_2 \in R_3(\forall r_2 \in R_3(\exists s_3 \in R_2(\forall r_3 \in R_2
\\&(s_1.1 \neq r_1.1 \lor s_2.1 \neq r_2.1 \lor s_3.1 \neq r_3.1 \\&\lor(r_3.3 = f.1 \land r_2.2 = r_3.1 \land r_1.2 = r_2.1 \land r_1.3 = f.2)))))))\\
q_5(f^{(1)}):=\; &\exists s_1 \in R_1(\forall r_1 \in R_1(\exists s_2 \in R_4(\forall r_2 \in R_4
\\&(s_1.1 \neq r_1.1 \lor s_2.1 \neq r_2.1 \lor s_2.2 \neq r_2.2 \lor (r_1.2 = r_2.1 \land r_1.3 = f.1)))))\\
q_6(f^{(1)}):=\; &\exists s_1 \in R_2(\forall r_1 \in R_2(\exists s_2 \in R_5(\forall r_2 \in R_5(\exists s_3 \in R_1(\forall r_3 \in R_1
\\&(s_1.1 \neq r_1.1 \lor s_2.1 \neq r_2.1 \lor s_3.1 \neq r_3.1 \\&\lor (r_2.1 = r_3.1 \land r_1.2 = r_2.2 \land r_1.2 = r_3.2 \land r_3.3 = f.1)))))))\\
q_7(f^{(1)}):=\; &\exists s_1 \in R_2(\forall r_1 \in R_2(\exists s_2 \in R_5(\forall r_2 \in R_5(\exists s_3 \in R_1(\forall r_3 \in R_1
\\&(s_1.1 \neq r_1.1 \lor s_2.1 \neq r_2.1 \lor s_3.1 \neq r_3.1 \\&\lor (r_1.2 = r_2.1 \land r_1.2 = r_3.1 \land r_1.1 = r_2.2 \land r_1.1 = r_3.2 \land r_3.3 = f.1)))))))\\
\end{align*}
\vspace{-0.3in}\\
All variables range over tuples in relations; $f^{(i)}$ denotes a tuple with $i$ elements, $i = 1, 2$.
\label{fo-rewritings}
\end{figure*}

\begin{figure*}[!ht]
\caption{The SQL translations of the ConQuer-FO-rewritings of the consistent answers to queries $q_1$ to $q_4$.}

\begin{minipage}[t]{0.45\textwidth}
$q_1(R1\_3):$\\
Candidates as (\\
select distinct R1.3 as R1\_3,R1.1 as R1\_1\\
from R1, R2\\
where (R1.2 = R2.1))\\

Filter as (select R1\_1 \\
from Candidates C join R1 on C.R1\_1 = R1.1 \\
left outer join R2 on R1.2 = R2.1 \\
where (R2.1 is null) \\
union all select C.R1\_1 from Candidates C \\
group by C.R1\_1 \\
having count(*) $>$ 1)\\

select R1\_3 \\
from Candidates C \\
where not exists \\
(select * from Filter F where C.R1\_1 = F.R1\_1)\\

$q_3(R1\_3):$\\
Candidates as (\\
select distinct R1.3 as R1\_3, R1.1 as R1\_1\\
from R1, R3, R2\\
where (R3.2 = R2.1) and (R1.2 = R3.1))\\

Filter as (\\
select R1\_1 \\
from Candidates C join R1 on C.R1\_1 = R1.1 \\
left outer join R3 on R1.2 = R3.1 left outer join R2 on R3.2 = R2.1 \\
where (R3.1 is null OR R2.1 is null) \\
union all select C.R1\_1 from Candidates C \\
group by C.R1\_1 \\
having count(*) $>$ 1)\\

select R1\_3 \\
from Candidates C \\
where not exists \\
(select * from Filter F where C.R1\_1 = F.R1\_1)\\
\end{minipage}
\hspace{0.1\textwidth}
\begin{minipage}[t]{0.45\textwidth}
$q_2(R1\_3, R2\_3):$\\
Candidates as (\\
select distinct R1.3 as R1\_3, R2.3 as R2\_3,\\
R1.1 as R1\_1 from R1, R2\\
where (R1.2 = R2.1))\\

Filter as (select R1\_1 \\
from Candidates C join R1 on C.R1\_1 = R1.1 \\
left outer join R2 on R1.2 = R2.1 \\
where (R2.1 is null) \\
union all select C.R1\_1 from Candidates C \\
group by C.R1\_1 \\
having count(*) $>$ 1)\\

select R1\_3, R2\_3 \\
from Candidates C \\
where not exists \\
(select * from Filter F where C.R1\_1 = F.R1\_1)\\

$q_4(R1\_3, R2\_3):$\\
Candidates as (\\
select distinct R1.3 as R1\_3, R2.3 as R2\_3,R1.1 as R1\_1\\
from R1, R3, R2\\
where (R3.2 = R2.1) and (R1.2 = R3.1))\\

Filter as (select R1\_1 \\
from Candidates C join R1 on C.R1\_1 = R1.1 \\
left outer join R3 on R1.2 = R3.1 left outer join R2 on R3.2 = R2.1 \\
where (R3.1 is null or R2.1 is null) \\
union all select C.R1\_1 from candidates C \\
group by C.R1\_1 \\
having count(*) $>$ 1)\\

select R1\_3, R2\_3 \\
from Candidates C \\
where not exists \\
(select * from Filter F where C.R1\_1 = F.R1\_1)\\
\end{minipage}

\label{conquer-fo-rewritings}
\end{figure*}

}
\noindent{\bf Results on the Real-World Databases}
In this set of experiments, we evaluated the performance of CAvSAT using Reduction \ref{reduction3} on real-world data having key constraints on each relation, along with one functional dependency. The data used are about  inspections of food establishments in New York and Chicago, and are taken from \cite{nydata} and \cite{chicagodata}. Part of this data have been previously used for evaluating  data cleaning systems, such as HoloClean \cite{Rekatsinas17}. Since the structure of the schema or the constraints on the database were not specified by the source, we decomposed the data into four relations, and assumed reasonable key constraints for all relations and also one additional functional dependency, as shown in Table \ref{real-data-schema}. We evaluated the performance of Reduction \ref{reduction3} on six queries depicted in Figure \ref{dc-queries} in the Appendix. For example, query $Q_3$ returns the names of the restaurants, such that they are present in both New York and Chicago, and they were inspected on the same day.  Figure \ref{real-queries} shows that the solver took the most amount of time to compute answers to this query. Not surprisingly, the evaluation time increases as the number of atoms or the number of free variables in the query grow. Table \ref{formula-size-realdata} shows the size of the CNF-formulas produced by Reduction \ref{reduction4} (where Reduction \ref{reduction3} is used inside Reduction \ref{reduction4}). No optimization was implemented in this set of experiments.
\movetoappendix{
\begin{figure*}[!ht]
\caption{Queries used on the real-world database.}
\begin{align*}
Q_1():=\; &\exists x, y, z, w, v, y', z', w', v' \; (\text{NY\_Rest}(x,y,z,w,v) \land \text{CH\_Rest}(x,y',z',w',v'))\\
Q_2(x):=\; &\exists y, z, w, v, y', z', w', v' \; (\text{NY\_Rest}(x,y,z,w,v) \land \text{CH\_Rest}(x,y',z',w',v'))\\
Q_3(x):=\; &\exists y, z, w, v, y', z', w', v', q, r, s, t, q', s', t' \; (\text{NY\_Rest}(x,y,z,w,v) \land \text{CH\_Rest}(x,y',z',w',v')\\ &\land \text{NY\_Insp}(y,q,r,s,t) \land \text{CH\_Insp}(y',q',r,s',t'))\\
Q_4(x, y):=\; &\exists z, w, v, q, r, s \; (\text{CH\_Rest}(x,y,z,w,v)\land \text{CH\_Insp}(y,q,r,s,\text{`Pass'}))\\
Q_5(x):=\; &\exists y, z, w, v, q, r, s \; (\text{CH\_Rest}(x,y,z,w,v)\land \text{CH\_Insp}(y,q,r,s,\text{`Fail'}))\; \cup\\
&\exists y, z, w, v, q, r, s \; (\text{NY\_Rest}(x,y,z,w,v)\land \text{NY\_Insp}(y,q,r,s,\text{`Fail'}))\\
Q_6(x, v):=\; &\exists y, z, w, y', z', w', v', q, r, s \; (\text{CH\_Rest}(x,y,z,w,v)\land \text{NY\_Rest}(x,y',z',w',v')\\&\land \text{NY\_Insp}(y',\text{`Not Critical'},q,r,s))\\
\end{align*}
\label{dc-queries}
\end{figure*}
}
\vspace{-0.3in}
\begin{table}[!ht]
\caption{The schema and the constraints of the real-world database.}
\centering
\begin{tabular}[t]{|l|l|}\hline
Relation & \# Tuples
\\\hline\hline
NY\_Insp (\underline{LicenseNo}, Risk, \underline{InspDate, InspType}, Result) & 229K \\\hline
NY\_Rest (Name, \underline{LicenseNo}, Cuisine, Address, Zip) & 26.5K\\\hline
CH\_Insp (\underline{LicenseNo}, Risk, \underline{InspDate, InspType}, Result) & 167K\\\hline
CH\_Rest (Name, \underline{LicenseNo}, Facility, Address, Zip) & 31.1K \\\hline
\end{tabular}
\vspace{0.08in}\\
\begin{tabular}[t]{|l|l|l|}\hline
Constraint & Type & Violations
\\\hline\hline
NY\_Insp (LicenseNo, InspDate, InspType $\rightarrow$ Risk, Result) & Key & 25.6\%\\\hline
NY\_Rest (LicenseNo $\rightarrow$ Name, Cuisine, Address, Zip) & Key & 0\%\\\hline
CH\_Insp (LicenseNo, InspDate, InspType $\rightarrow$ Risk, Result) & Key & 0.07\%\\\hline
CH\_Rest (LicenseNo $\rightarrow$ Name, Cuisine, Address, Zip) & Key & 5.86\%\\\hline
CH\_Rest (Name $\rightarrow$ Zip) & FD & 9.73\%\\\hline
\end{tabular}
\label{real-data-schema}
\end{table}
\vspace{-0.4in}
\begin{table}[!ht]
\begin{minipage}[t]{0.55\linewidth}

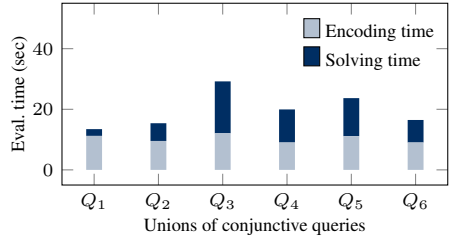
\captionof{figure}{Evaluation time of CAvSAT on real data.}
    \centering
    \vspace{0.12in}
    \begin{tikzpicture}
\begin{axis}[ybar stacked,
		width=\textwidth,height=4cm,
		xlabel= Unions of conjunctive queries,
		label style={font=\scriptsize},
		xtick=data,
		xticklabels={$Q_1$,$Q_2$,$Q_3$,$Q_4$,$Q_5$,$Q_6$},
		legend style={at={(1,0.95)},font=\scriptsize,draw=none,fill=none},
		legend cell align = {left},
		bar width=6pt,
		ymin=0, ymax=50,
		enlargelimits=0.1,
		y label style={at={(0.13,0.5)}},
		x label style={at={(0.5,0.08)}},
		xmin = 1,xmax=6,
		ylabel=Eval. time (sec)]
	\addplot[fill=coolblack!30!white,draw=none] coordinates {
	(1, 11.216)
		(2,9.544) 
		(3,12.079)
		(4,9.035)
		(5,11.128)
		(6, 9.034)
	};\addlegendentry{Encoding time}
	\addplot[fill=coolblack, draw=none] coordinates {
	(1, 2.235)
		(2,5.854)
		(3,17.130)
		(4,10.915)
		(5,12.554)
		(6, 7.437)
	};\addlegendentry{Solving time}
	\end{axis}
\end{tikzpicture}
    \label{real-queries}
  \end{minipage}
  \hspace{2em}
  \begin{varwidth}[t]{0.45\linewidth}
  \caption{Size of the CNF-formula.}
  \begin{center}
  \renewcommand{\arraystretch}{1.1}
\setlength{\tabcolsep}{3pt}
    \begin{tabular}[h]{|c|c|c|}\hline
Query & Variables & Clauses\\\hline\hline
$Q_1$ & 455.1K & 793.7K\\\hline
$Q_2$ & 456.5K& 794K \\\hline
$Q_3$ &455.1K &671.5K \\\hline
$Q_4$ & 476K & 861.5K\\\hline
$Q_5$ & 486.7K & 836.2K\\\hline
$Q_6$ & 455.5K & 1.12M\\\hline
\end{tabular}\label{formula-size-realdata}
\end{center}
  \end{varwidth}
\end{table}
\newpage
\section{Concluding Remarks}
We designed and implemented CAvSAT, the first SAT-based system for consistent query answering. Our preliminary stand-alone evaluation shows that a SAT-based approach can give rise to a scalable system for consistent query answering. We note that, on queries with first-order rewritable consistent answers, CAvSAT had comparable or even better performance to evaluating the first-order rewritings using a database engine. This finding suggests a potential difference between theory and practice, since the study of first-order rewritability of the consistent answers was motivated from having an efficient evaluation of consistent answers using the database engine alone.

The next step in this investigation is to carry out an extensive comparative evaluation of CAvSAT with other systems for consistent query answering and, in particular, with systems that use reduction-based methods \cite{Kolaitis13,MannaRT11}.

\smallskip

\noindent{\bf Acknowledgments} Dixit is supported by the Center for Research in Open Source Software (CROSS) at UC Santa Cruz. Kolaitis  is  supported by NSF Grant IIS:1814152.

\clearpage
\bibliography{cqa-references,dataexchange,bib}
\clearpage
\appendix
\section*{Appendix}
\immediate\closeout\appendixfile\input{auto-appendix.tex}
\end{document}